\documentclass[lettersize, journal]{IEEEtran}
\usepackage{float}
\usepackage{caption}

\usepackage{color}

\usepackage{cite}
\usepackage{amsmath,amssymb,amsfonts}
\usepackage{amsthm}
\usepackage{graphicx}
\usepackage{textcomp}
\usepackage{xcolor}
\usepackage{svg}
\usepackage{subcaption}

\usepackage{mathtools}

\usepackage[ruled,vlined]{algorithm2e}

\usepackage{epstopdf}
\usepackage{multirow}
\usepackage{pifont}

\usepackage{stfloats}

\newtheorem{remark}{Remark}
\newtheorem{theorem}{Theorem}

\newtheorem{lemma}{Lemma}

\newtheorem{corollary}{Corollary}

\allowdisplaybreaks
\makeatletter
\def\ScaleIfNeeded{%
\ifdim\Gin@nat@width>\linewidth \linewidth \else \Gin@nat@width
\fi } \makeatother

\graphicspath{{Fig/}}

\begin{document}

\title{Pinching-Antenna Systems-enabled Secure ISAC: A Two-Timescale Optimization Framework}
\author{Haowen Song, Jingjing Zhao,~\IEEEmembership{Senior Member,~IEEE}, Xidong Mu, and Kaiquan Cai 
\thanks{H. Song, J. Zhao, and K. Cai are with the School of Electronics and Information Engineering, Beihang University, 100191, Beijing, China, and also with the State Key Laboratory of CNS/ATM, 100191, Beijing, China. (e-mail:\{haowensong, jingjingzhao, caikq\}@buaa.edu.cn). 

X. Mu is with the Centre for Wireless Innovation (CWI), Queen's University Belfast, Belfast, BT3 9DT, U.K. (e-mail: x.mu@qub.ac.uk). }}
\maketitle
\begin{abstract}
A novel two-timescale optimization framework is proposed for pinching-antenna systems (PASS)-enabled secure integrated sensing and communications (ISAC). Specifically, a base station (BS) equipped with pinching antennas (PAs) transmits signals to a legitimate user under the existence of an eavesdropper (Eve), while employing leaky coaxial cables (LCXs) for receiving echo signals to track Eve’s mobility
states, i.e., locations and velocities. Considering the practical PAs activation overhead, the pinching beamforming and baseband processing are optimized in the large and small timescales, respectively. The multiple-waveguide scenario is first considered, where the BS can transmit the artificial noise together with communication signals for both jamming and sensing purposes. A joint baseband and pinching beamforming design problem is formulated to maximize the average secrecy rate. To address this problem, an alternating optimization algorithm is first invoked for jointly optimizing the pinching and baseband beamforming with predicted Eve’s mobility states. With determined PAs positions, the baseband beamforming is updated with refined Eve’s states obtained from real-time echo signal processing. The single-waveguide scenario is then considered. Since a single waveguide carries at most one independent data stream, an ISAC framework with separate communication and sensing phases is proposed. The element-wise algorithm proposed for the multiple-waveguide scenario is extended to solve the resultant pinching beamforming problem. Numerical results demonstrate that: 1) Eve’s velocities and positions can be accurately tracked with the proposed two-timescale framework in both multiple- and single-waveguide scenarios; and 2) PASS achieves superior secrecy rate compared to conventional multiple-antenna benchmarks.

\end{abstract}
\begin{IEEEkeywords}
Pinching-antenna systems, pinching beamforming, physical layer security, integrated sensing and communications.
\end{IEEEkeywords}
\section{Introduction}
Wireless communication systems have witnessed rapid advancements in recent years, while the performance requirements imposed on them have become increasingly diverse and demanding. As a consequence, next-generation communication technologies are expected to support not only high-rate data transmission but also additional capabilities such as environmental sensing. Such communication and sensing capabilities are particularly valuable for Internet of Things applications, such as intelligent transportation and industrial automation, which require reliable connectivity and real-time environmental awareness~\cite{9606831}. Driven by these requirements, the sensing capability of wireless signals has attracted growing attention from both academia and industry, paving the way for the development of integrated sensing and communication (ISAC) networks~\cite{9737357}. ISAC enables the integration of communication and sensing functionalities over shared wireless resources, thereby enhancing spectrum utilization and improving overall system efficiency~\cite{9755276}. 

Beyond resource sharing, ISAC also provides a new opportunity to endow wireless networks with environmental awareness. In communication-centric scenarios, such awareness can be exploited to facilitate communication-oriented designs, rather than treating sensing as an independent objective. This gives rise to the sensing-assisted communication paradigm, where sensing serves as an auxiliary function to support beamforming, resource allocation, and secure transmission.  Specifically, this paradigm leverages sensing-derived environmental knowledge, such as location, velocity, angle, and channel-related parameters, to guide subsequent communication design. It has been widely investigated for communication tasks such as beam training~\cite{10609801}, beam tracking~\cite{9171304}, and channel estimation~\cite{10845870}. Moreover, this capability is particularly valuable in physical layer security (PLS) systems, where obtaining the channel state information (CSI) of eavesdroppers (Eves) at the transmitter is generally difficult because Eves are typically non-cooperative nodes~\cite{9014513}. Motivated by this, sensing-aided communication has recently emerged as a promising approach for acquiring wiretap CSI and thereby facilitating subsequent beamforming design~\cite{wei2023integrated, 9829746}. Compared with conventional pilot-based estimation methods, this approach eliminates the requirement for transmitting pilot signals, and instead estimating the CSI based on the parameters extracted from echo signals reflected by Eves~\cite{10418473}.

Recently, PASS has emerged as a novel flexible-antenna architecture, whose first prototype was demonstrated by NTT DOCOMO~\cite{docomo}. PASS employs dielectric waveguides with low in-waveguide propagation loss as the transmission medium, and its aperture length spans from a few meters to tens of meters. Along the waveguide, small dielectric elements, referred to as pinching antennas (PAs), can be dynamically deployed to transmit radio waves~\cite{yue2026performance}. This scalable architecture not only supports a flexible line-of-sight (LoS) link establishment but also significantly reduces the path loss, thereby creating new opportunities for enhancing secure communications~\cite{10945421, 11169486, 11215676}. These unique characteristics also make PASS particularly appealing for sensing-aided secure communications. From a communication perspective, PASS establishes “near-wired” links with stable LoS conditions to legitimate users, which effectively mitigates free-space path loss and avoids LoS blockage~\cite{11342384}. From a sensing perspective, the long waveguide synthesizes a large effective aperture, which induces dominant near-field effects that facilitate precise localization~\cite{10934790}. Motivated by these advantages, the sensing-assisted secure communications in PASS deserves further investigations.

\subsection{Related Works}
In recent years, growing research interests have been devoted to sensing-aided secure communication. For instance, the authors of~\cite{9968163} investigated PLS in an ISAC-enabled radar-communication coexistence system and proposed joint transmit beamforming designs to enhance secrecy performance when the Eves’ CSI is unavailable. In~\cite{10153696}, the authors studied robust beamforming optimization for downlink secure ISAC systems, where secure communication and sensing performance were well balanced by jointly designing information and sensing beams under imperfect CSI of potential Eves. In this work, artificial noise (AN) employed at the base station serves a dual role: degrading the reception quality of Eves and simultaneously serving as a probing waveform for sensing. By reusing AN for both jamming and sensing, the sensing requirement can be fulfilled while preserving the secrecy performance. Moreover, the authors of~\cite{9199556} investigated PLS in a MIMO dual-functional radar-communication (DFRC) system and proposed AN-aided beamforming designs to suppress eavesdropping at radar targets while guaranteeing the signal-to-interference-plus-noise ratio (SINR) requirements of legitimate users. As a further advance, in~\cite{10227884}, a hybrid Capon and approximate maximum likelihood approach was adopted to infer the directions of potential Eves, and a weighted optimization framework was developed to jointly enhance the secrecy rate and reduce the Eve's Cramér-Rao bound (CRB). Furthermore, sensing-assisted predictive beamforming has attracted considerable attention in high-mobility scenarios. In~\cite{10345500}, a secure unmanned aerial vehicle (UAV) communication framework was proposed against a mobile Eve, where AN was exploited to sense the Eve and estimate the wiretap channel, so as to assist the online UAV navigation and resource allocation design.

Recently, PASS-enabled secure communications have also attracted growing research attentions. In~\cite{sun2025physical}, the authors proposed gradient-based and fractional programming algorithms for secure beamforming in PASS, where the baseband and pinching beamforming were jointly optimized to maximize the secrecy rate in both single- and multi-user scenarios. Moreover, a coalitional game-based pinching beamforming algorithm was developed in~\cite{11215679}, where PAs positions were designed to enhance the legitimate user’s signal quality through amplitude control of the legitimate user’s effective channel, while degrading the Eve’s reception via phase alignment. Furthermore, the authors of~\cite{11202497} proposed a PASS-enabled secure communication framework, where a low-complexity PA-wise successive tuning algorithm was developed for the single-waveguide scenario, and both waveguide division (WD) and waveguide multiplexing (WM) structures were studied in the multiple-waveguide scenario to improve secrecy performance. However, the aforementioned works~\cite{sun2025physical,11215679,11202497} all assumed that the wiretap CSI is well known at the transmitter.

\subsection{Motivations and Contributions}
Existing studies on PASS-enabled secure communications assume that the wiretap CSI is known and that Eve remains stationary, which limits their practical applicability to dynamic secure communication scenarios. To address the above issues, in this paper, we leverage the ISAC functionality of PASS to track the CSI of a moving malicious Eve. It is worthy noting that, due to the non-negligible latency and overhead introduced by PAs activation, we need to consider the practical limitation that the pinching beamforming cannot adapt to Eve’s mobility in real time~\cite{11364174}. To this end, we propose a novel two-timescale joint pinching beamforming and baseband processing framework for the PASS-enabled secure ISAC. Specifically, the pinching beamforming and baseband processing are performed in the in the large timescale (i.e., a time block containing multiple coherent processing intervals (CPIs)), and the small timescale (i.e., a CPI), respectively. The main contributions of this
work are summarized as follows.

\begin{itemize}
    \item We consider the scenario of sensing-aided secure communications empowered by PASS, where a base station (BS) is deployed with PAs for signal transmissions to a legitimate user, under the existence of a moving eavesdropper (Eve). To track the Eve's wiretap CSI, the BS adopts leaky coaxial cables (LCXs) for estimating Eve's velocities and positions with echo signals, by leveraging the near-field characteristics given the extended receiving aperture. Considering the overhead associated with PAs positions reconfiguration, we propose to optimize the pinching beamforming for a predefined time period, (i.e., a time block), while allowing baseband processing to be executed in real time, thereby facilitating a two-timescale optimization framework.
    \item We first consider the multiple-waveguide scenario, where the BS is enabled to transmit the AN together with communication signals, for both sensing and jamming purposes. We formulate the average secrecy rate maximization problem in each time block, i.e., the time duration in which PAs positions keep unchanged, subject to the sensing performance constraint. To solve the resultant optimization problem across two different timescales, the alternating optimization (AO) algorithm is first invoked for jointly solving the pinching and baseband beamforming problem with the predicted Eve’s mobility states, which is followed by the real-time refinement of the baseband beamforming with the continuously-updated Eve’s states. 
    \item We then consider the single-waveguide scenario. Since a single waveguide supports at most one independent data stream, we design the ISAC frame structure consisting of separate communication and sensing phases in each time block. Accordingly, the joint pinching beamforming and power allocation problem is formulated for maximizing the average secrecy rate in the communication phase, under the constraint of minimum echo signal power in the sensing phase. The large-timescale pinching beamforming is optimized in an element-wise manner, while the small-timescale power allocation is adaptively determined according to the effective channel quality comparison between Bob and Eve. 
    \item Numerical results unveil that: 1) In both multiple- and single-waveguide scenarios, Eve's velocities and positions can be tracked with high precision, by employing the proposed two-timescale optimization framework; and 2) Under the PAs activation overhead limitation, PASS still outperforms both conventional MIMO and massive MIMO schemes in terms of secrecy rate with the assistance of ISAC. 

\end{itemize}

\subsection{Organization and Notations}
The rest of this paper is structured as follows. Section II studies the two-timescale ISAC-enabled secure communication system model for multiple-waveguide. Section III proposes the Eve's trajectory tracking method, followed by the joint pinching and baseband beamforming problem formulation. Section IV proposes the two-timescale optimization algorithm for the multiple-waveguide scenario. Section V studies the single-waveguide scenario and develops the corresponding joint pinching beamforming and power allocation designs. Section VI presents numerical results and Section VII concludes the paper.

$\textit {Notations}$: Italic letters, bold-face lower-case letters, and bold-face upper-case letters denote scalars, vectors, and matrices, respectively. $\mathbb{C}^{N\times M}$ denotes the set of $N\times M$ complex-valued matrices. Superscripts $(\cdot)^T$, $(\cdot)^H$, and $(\cdot)^{-1}$ denote the transpose, conjugate transpose, and inverse, respectively. $|\cdot|$ and $\left\|\cdot\right\|$ denote the modulus of a scalar and the Euclidean norm of a vector, respectively. $\text{Tr}\left(\cdot\right)$ denotes the trace of a matrix. Moreover, $\odot$ and $\left\langle\cdot,\cdot\right\rangle$ denote the Hadamard product and inner product, respectively. Unless otherwise specified, all random variables are assumed to be zero mean.

\section{Two-Timescale PASS-enabled Secure ISAC with Multiple Waveguides}
As illustrated in Fig.~1(a), we consider a sensing-aided secure communication system consisting of a BS, a single-antenna legitimate user, referred to as Bob, and a single-antenna Eve. Due to the Eve’s malicious nature, it is challenging for the BS to obtain the wiretap CSI. To address this issue, ISAC technology is utilized to enable reliable communication with Bob while simultaneously leveraging its sensing capability to obtain the Eve’s CSI. The BS is deployed with $N_{\text{t}}$ dielectric waveguides, each of which is fed by a dedicated RF chain and contains $M_{\text{t}}$ PAs. Let $\mathcal{N}_{\text{t}} = \left\{1,\dots, N_{\text{t}}\right\}$ and $\mathcal{M}_{\text{t}}^n=\left\{1,\dots, M_{\text{t}}\right\}$ denote the set of waveguides and that of PAs on the $n$-th waveguide, respectively. Moreover, $N_{\text{r}}$ leaky coaxial cables (LCXs) are employed at the BS for echo reception~\cite{11111701}, each containing $M_{\text{r}}$ uniformly spaced slots for receiving echo signals. Both the transmit waveguides and receive LCXs are assumed to be aligned parallel to the $x$-axis at a height of $d$. Furthermore, the lengths of both the waveguides and the LCXs are set to $L$. Let $\bar{\boldsymbol{\psi}}^{n}_{\text{p}}=[0,y_{\text{p}}^{n},d]$ and $\boldsymbol{\psi}_{\text{p}}^{n,m} = \left[x_{\text{p}}^{n,m}, y_{\text{p}}^{n}, d\right]$ denote the positions of the feed point and the $m$-th PA on the $n$-th waveguide, respectively. The $x$-axis coordinates of PAs deployed on the $n$-th waveguide are collected in $\mathbf{x}_{\text{p}}^{n} = \left[x_{\text{p}}^{n,1}, ..., x_{\text{p}}^{n,M_{\text{t}}}\right]\in\mathbb{R}^{1\times M_{\text{t}}}$. Accordingly, the PAs positions matrix over all waveguides is defined as $\mathbf{X}_{\text{p}} = \left[\left(\mathbf{x}_{\text{p}}^{1}\right)^T, ..., \left(\mathbf{x}_{\text{p}}^{N_{\text{t}}}\right)^T\right]^T\in\mathbb{R}^{N_{\text{t}}\times M_{\text{t}}}$. Assume that PAs on each waveguide are placed in a successive order, i.e., $x_{\text{p}}^{n,{m+1}}>x_{\text{p}}^{n,{m}}, \forall 1\leq m< M_{\text{t}}, \forall n$, and the maximum deployment range of PAs is $L$. We assume that the Eve moves on the $xy$-plane, while Bob remains stationary. Time is slotted into coherent processing intervals (CPIs) indexed by $t$ with the duration of $\Delta T$, during which the Eve’s location and velocity are assumed to be unchanged. Denote the location of Bob, the location and velocity of Eve in the $t$-th CPI by $\boldsymbol{\psi}_{\text{b}} = \left[x_{\text{b}}, y_{\text{b}}, 0\right]$, $\boldsymbol{\psi}_{\text{e}}^{t} = \left[x_{\text{e}}^{t}, y_{\text{e}}^{t}, 0\right]$ and $\mathbf{v}_{\mathrm{e}}^{t}=[v_{x}^{t},v_{y}^{t},0]^{T}$, respectively.  For improving the sensing and secrecy performance, the BS inserts AN into the transmit signals for both sensing and jamming purpose. Then, the transmit signal at the BS in the $t$-th CPI can be expressed as
\begin{align}
\mathbf{c}^{t}=\mathbf{w}^{t}s+\mathbf{z}^{t},
\end{align}
where $s\in\mathbb{C}$ denotes the information signal for Bob with $\mathbb{E}\{|s|^2\}=1$, $\mathbf{w}^{t}\in\mathbb{C}^{N_{\text{t}}\times1}$ denotes the baseband beamforming vector and $\mathbf{z}^{t}\in\mathbb{C}^{N_{\text{t}}\times1}$ denotes the AN vector. 

\begin{figure}
    \centering
    \begin{subfigure}{\linewidth}
        \centering
        \includegraphics[scale=0.3]{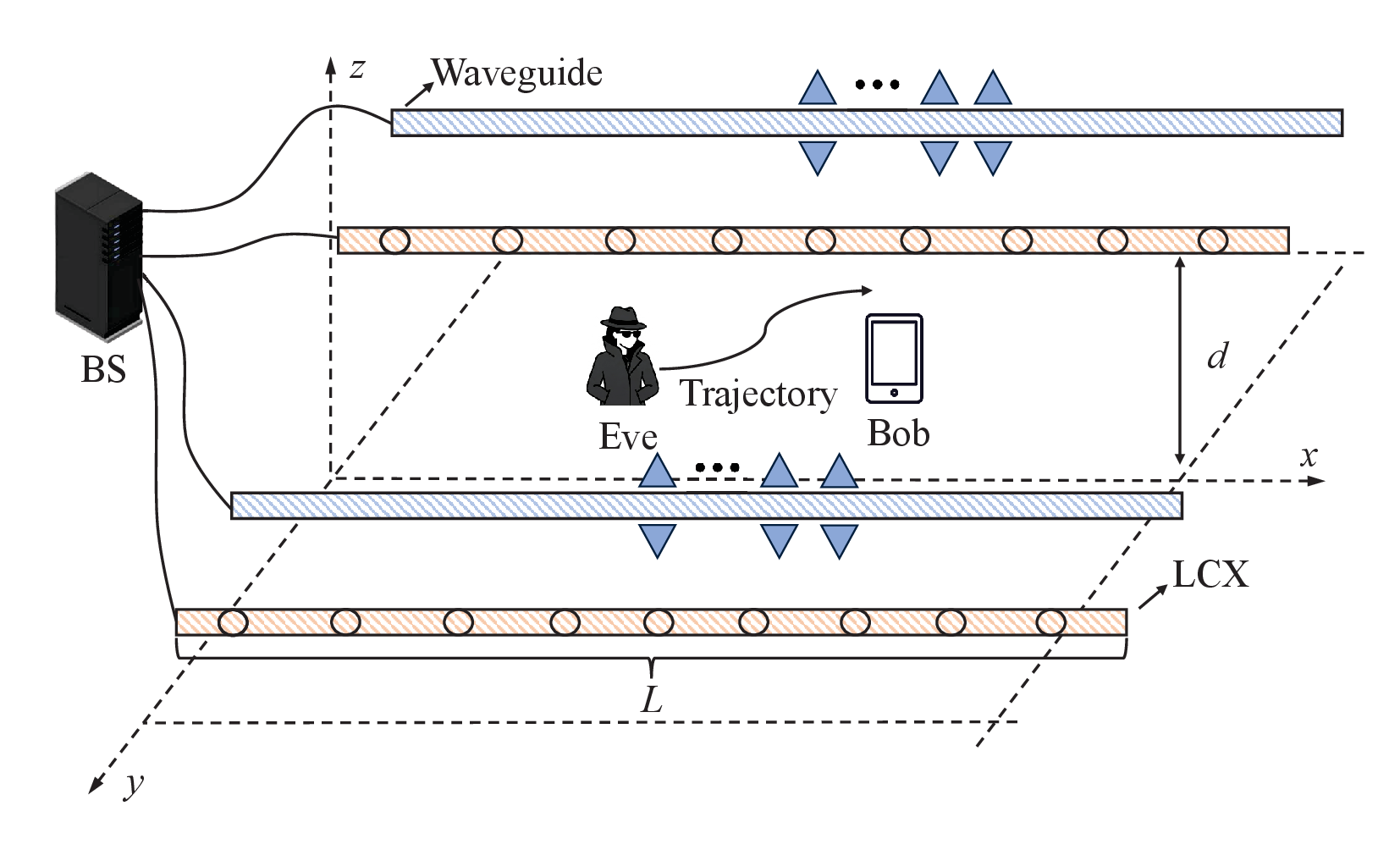}
        \caption{}
        \label{fig:system_model}
    \end{subfigure}
    \begin{subfigure}{\linewidth}
        \centering
        \includegraphics[scale=0.3]{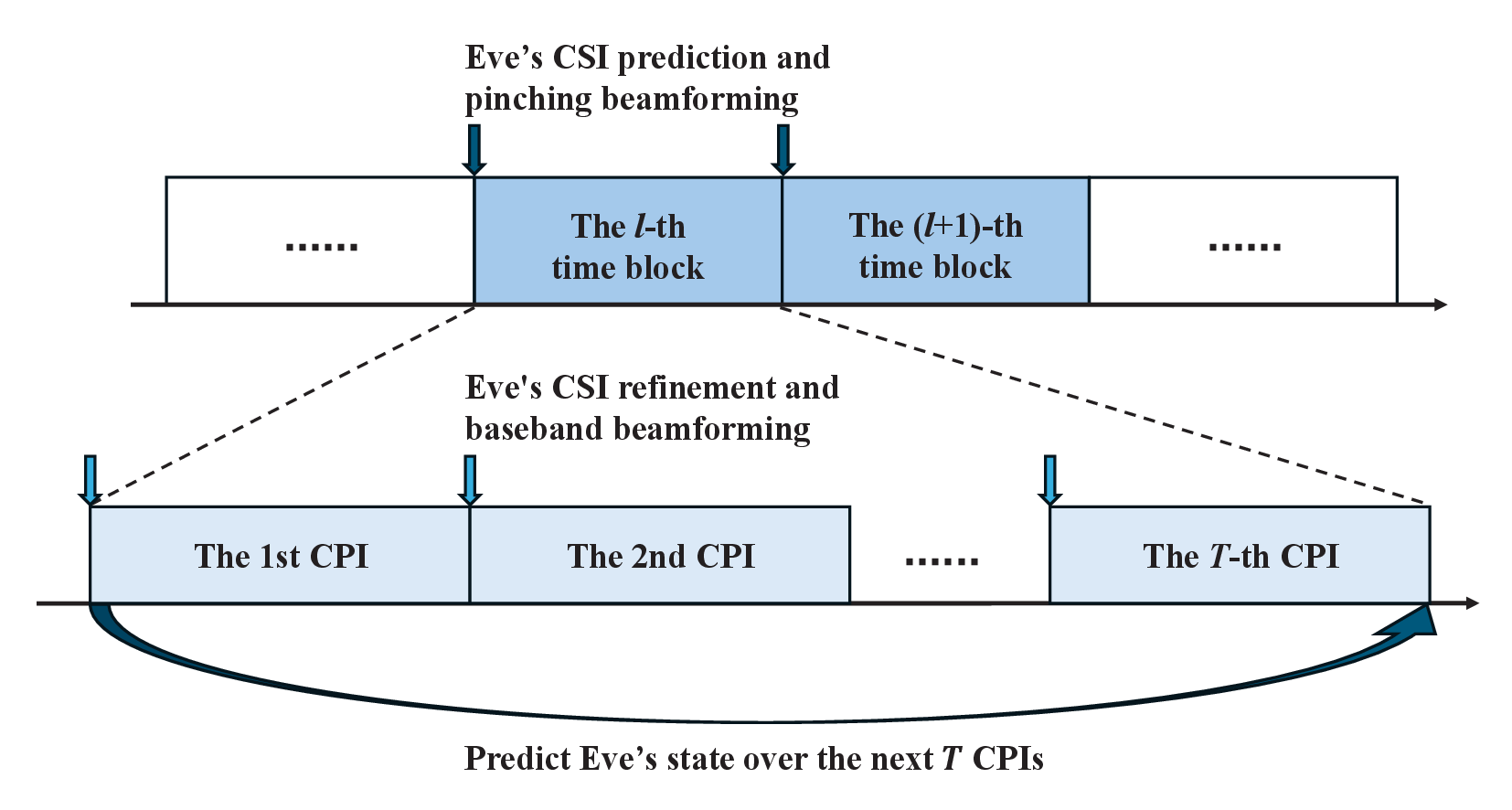}
        \caption{}
        \label{fig:frame}
    \end{subfigure}
    \caption{Illustration of (a) the system model of the ISAC-enabled secure communications in multiple-waveguide PASS; and (b) two-timescale pinching and baseband beamforming framework.}
    \label{fig:system_model_frame}
\end{figure}


\begin{remark}
Due to the latency and overhead caused by PAs activation, PAs positions cannot be adjusted in real time to adapt to Eve’s mobility. To address this issue, we propose a two-timescale beamforming framework, where the pinching and baseband beamforming are performed over the large and small timescale, respectively. As illustrated in Fig. 1(b), on the one hand, PAs positions are updated at the beginning of each time block consisting of $T$ CPIs with the predicted Eve's CSI. On the other hand, the baseband beamforming is updated at each CPI to adapt to channel variations.
\end{remark}


\subsection{Sensing Signal Model}
The round-trip channel between the BS and Eve in the $t$-th CPI includes both the forward probing and the reflected channels. Since the pinching beamforming is optimized over the large timescale and remains unchanged within the time block, the time index $t$ of $\mathbf{X}_{\text{p}}$ is omitted for notational simplicity. The signal propagation response $\mathbf{g}\left(\mathbf{x}_{\text{p}}^{n}\right)\in\mathbb{C}^{M_{\text{t}}\times 1}$ from the feed point to PAs over the $n$-th waveguide is given by
\begin{align}
\label{eq:multi-waveguide-channel}
    \mathbf{g}\left(\mathbf{x}_{\text{p}}^{n}\right) & = \frac{1}{\sqrt{M_{\text{t}}}}\left[e^{-j\frac{2\pi \left\|\boldsymbol{\psi}_{\text{p}}^{n,1}-\bar{\boldsymbol{\psi}}_{\text{p}}^{n}\right\|}{\lambda_{\text{g}}}},..., e^{-j\frac{2\pi \left\|\boldsymbol{\psi}_{\text{p}}^{n,M_{\text{t}}}-\bar{\boldsymbol{\psi}}_{\text{p}}^{n}\right\|}{\lambda_{\text{g}}}}\right]^T\nonumber\\
    &=\frac{1}{\sqrt{M_{\text{t}}}}\left[e^{-j\frac{2\pi x_{\text{p}}^{n,1}}{\lambda_{\text{g}}}},..., e^{-j\frac{2\pi x_{\text{p}}^{n,M_{\text{t}}}}{\lambda_{\text{g}}}}\right]^T,
\end{align}
where $\lambda_{\text{g}}=\frac{\lambda}{n_{\text{eff}}}$ is the guided wavelength, with $\lambda$ and $n_{\text{eff}}$ being the free-space wavelength and the effective refractive index of the dielectric waveguide, respectively. Following~\cite{11202577}, the transmit power assigned to each waveguide is assumed to be uniformly radiated by its $M_{\text{t}}$ PAs, and thus the normalization factor $\frac{1}{\sqrt{M_{\text{t}}}}$ is included in the propagation response. In addition, the in-waveguide propagation loss is neglected in Eq.~\eqref{eq:multi-waveguide-channel}, since it is much smaller than the free-space path loss, as verified in~\cite{kaidi}. Then, the array response vector between PAs on the $n$-th waveguide and Eve at the beginning of the $t$-th CPI is given by
\begin{align}
    & \boldsymbol{\alpha}_{\text{e,t}}^{n,t}\left(\boldsymbol{\psi}_{\text{e}}^{t},\mathbf{x}_{\text{p}}^{n}\right)= \nonumber\\
    &\left[\frac{\eta e^{-j\frac{2\pi}{\lambda}\left\|\boldsymbol{\psi}_{\text{e}}^{t}-\boldsymbol{\psi}^{n,1}_{\text{p}}\right\|}}{\left\|\boldsymbol{\psi}_{\text{e}}^{t}-\boldsymbol{\psi}^{n,1}_{\text{p}}\right\|}, ..., \frac{\eta e^{-j\frac{2\pi}{\lambda}\left\|\boldsymbol{\psi}_{\text{e}}^{t}-\boldsymbol{\psi}^{n,M_{\text{t}}}_{\text{p}}\right\|}}{\left\|\boldsymbol{\psi}_{\text{e}}^{t}-\boldsymbol{\psi}^{n,M_{\text{t}}}_{\text{p}}\right\|}\right],
    \label{eq:multi-h-channel-Eve-pro}
\end{align}
where $\eta=\frac{\lambda}{4\pi}$ represents the channel gain at the reference distance of $1$~m, and $\left\|\boldsymbol{\psi}_{\text{e}}^{t}-\boldsymbol{\psi}^{n,m}_{\text{p}}\right\|$ is given by
\begin{align}
\left\|\boldsymbol{\psi}_{\text{e}}^{t}-\boldsymbol{\psi}^{n,m}_{\text{p}}\right\|=\sqrt{\left(x_\text{e}^{t}-x_\mathrm{p}^{n,m}\right)^2+\left(y_\text{e}^{t}-y_\mathrm{p}^{n}\right)^2+d^2}.    
\end{align}

Owing to the near-field effects induced by the large transmitting/receiving apertures, the Doppler shifts become non-uniform across all PAs and LCX slots. Specifically, we first define the vector from the $m$-th PA on the $n$-th waveguide to Eve in the $t$-th CPI as
\begin{align}
\mathbf{r}^{n,m,t}_{\text{t}}\triangleq\boldsymbol{\psi}_{\text{e}}^{t}-\boldsymbol{\psi}^{n,m}_{\text{p}}.
\end{align}
Then, the direction of $\mathbf{r}^{n,m,t}_{\text{t}}$ can be given by $\hat{\mathbf{r}}^{n,m,t}_{\text{t}}=\mathbf{r}^{n,m,t}_{\text{t}}/\|\mathbf{r}^{n,m,t}_{\text{t}}\|_2$. Accordingly, the projected velocity of Eve onto the direction $\hat{\mathbf{r}}^{n,m,t}_{\text{t}}$ is given by
\begin{align}
v^{n,m,t}_{\text{t}}=\left\langle\hat{\mathbf{r}}^{n,m,t}_{\text{t}},\mathbf{v}_{\mathrm{e}}^{t}\right\rangle=\hat{\mathbf{r}}^{n,m,t}_{\text{t}}\mathbf{v}_{\mathrm{e}}^{t},
\end{align}
which is the Eve's radial velocity relative to the $m$-th PA on the $n$-th waveguide. Thus, considering all PAs on the $n$-th waveguide, the time-variant phase-Doppler vector $\mathbf{d}_{\text{t}}^{n,t}\in\mathbb{C}^{1\times M_{\text{t}}}$ is given by
\begin{align}
\mathbf{d}_{\text{t}}^{n,t}\left(\mathbf{v}_{\mathrm{e}}^{t},\boldsymbol{\psi}_{\text{e}}^{t},\mathbf{x}_{\text{p}}^{n}\right)=\left[e^{-\mathrm{j}2\pi f^{n,1,t}_{\text{t}}\Delta T},\ldots,e^{-\mathrm{j}2\pi f^{n,M_{\text{t}},t}_{\text{t}}\Delta T}\right],
\end{align}
where $f^{n,m,t}_{\text{t}}=\frac{v^{n,m,t}_{\text{t}}}{\lambda}$ denotes the Doppler frequency shift over the $m$-th antenna on the $n$-th waveguide in the $t$-th CPI. As such, the wireless channel vector from PAs on the $n$-th waveguide to Eve is given by
\begin{align}
\label{eq:multi-overall-channel-pro}
\mathbf{h}_{\text{e,t}}^{n,t}\left(\mathbf{v}_{\mathrm{e}}^{t},\boldsymbol{\psi}_{\text{e}}^{t},\mathbf{x}_{\text{p}}\right)\triangleq\boldsymbol{\alpha}_{\text{e,t}}^{n,t}\left(\boldsymbol{\psi}_{\text{e}}^{t},\mathbf{x}_{\text{p}}\right)\odot\mathbf{d}_{\text{t}}^{n,t}\left(\mathbf{v}_{\mathrm{e}}^{t},\boldsymbol{\psi}_{\text{e}}^{t},\mathbf{x}_{\text{p}}\right).    
\end{align}
Further denote $\mathbf{h}_{\text{e,t}}^{t}\left(\mathbf{v}_{\mathrm{e}}^{t},\boldsymbol{\psi}_{\text{e}}^{t},\mathbf{X}_{\text{p}}\right)=\left[\mathbf{h}_{\text{e,t}}^{1,t}\left(\mathbf{v}_{\mathrm{e}}^{t},\boldsymbol{\psi}_{\text{e}}^{t},\mathbf{x}_{\mathrm{p}}^{1}\right),\ldots,\mathbf{h}_{\text{e,t}}^{N_{\text{t}},t}\left(\mathbf{v}_{\mathrm{e}}^{t},\boldsymbol{\psi}_{\text{e}}^{t},\mathbf{x}_{\mathrm{p}}^{N_{\text{t}}}\right)\right]\in\mathbb{C}^{1\times N_{\text{t}}M_{\text{t}}}$ as the channel vector from PAs over all waveguides to Eve. Therefore, the
overall probing signal arriving at Eve in the $t$-th CPI is given by
\begin{align}
\label{eq:multi-signal-Eve}
y_\text{e}^{t}=\mathbf{h}_{\text{e,t}}^{t}\left(\mathbf{v}_{\mathrm{e}}^{t},\boldsymbol{\psi}_{\text{e}}^{t},\mathbf{X}_{\text{p}}\right)\mathbf{G}\left(\mathbf{X}_{\text{p}}\right)\mathbf{c}^{t}+n_\mathrm{e}, 
\end{align}
where $\mathbf{G}\left(\mathbf{X}_{\text{p}}\right)\in\mathbb{C}^{N_{\text{t}}M_{\text{t}}\times N_{\text{t}}}$ denotes the propagation response from feed points to PAs over all waveguides, given by
\begin{align}
\mathbf{G}\left(\mathbf{X}_{\text{p}}\right)=
\begin{bmatrix}
\mathbf{g}\left(\mathbf{x}_\mathrm{p}^1\right) & \mathbf{0} & \ldots & \mathbf{0} \\
\mathbf{0} & \mathbf{g}\left(\mathbf{x}_\mathrm{p}^2\right) & \ldots & \mathbf{0} \\
\vdots & \vdots & \ddots & \vdots \\
\mathbf{0} & \mathbf{0} & \ldots & \mathbf{g}\left(\mathbf{x}_\mathrm{p}^{N_{\text{t}}}\right)
\end{bmatrix},
\end{align}
and $n_{\mathrm{e}}\sim\mathcal{CN}(0,\sigma_{\mathrm{e}}^{2})$ denotes the additive white Gaussian noise (AWGN) at Eve, with zero mean and a power of $\sigma_{\mathrm{e}}^{2}$. 

The probing signal will be reflected at Eve and captured by LCXs for sensing. Denote the position of the $m$-th slot on the $q$-th LCX by $\boldsymbol{\psi}_{\text{c}}^{q,m} = \left[x_{\text{c}}^{q,m}, y_{\text{c}}^{q}, d\right]$. The vector from $\boldsymbol{\psi}_{\text{e}}^{t}$ to $\boldsymbol{\psi}_{\text{c}}^{q,m}$ is given by
\begin{align}
\mathbf{r}^{q,m,t}_{\text{r}}\triangleq\boldsymbol{\psi}_{\text{e}}^{t}-\boldsymbol{\psi}^{q,m}_{\text{c}},
\end{align}
whose direction can be given by $\hat{\mathbf{r}}^{q,m,t}_{\text{r}}=\mathbf{r}^{q,m,t}_{\text{r}}/\|\mathbf{r}^{q,m,t}_{\text{r}}\|_2$. Thus, the Eve's radial velocity with respect
to the $m$-th slot on the $q$-th LCX is given by
\begin{align}
v^{q,m,t}_{\text{r}}=\left\langle\hat{\mathbf{r}}^{q,m,t}_{\text{r}},\mathbf{v}_{\mathrm{e}}^{t}\right\rangle=\hat{\mathbf{r}}^{q,m,t}_{\text{r}}\mathbf{v}_{\mathrm{e}}^{t}.
\end{align}
Considering all slots on the $q$-th LCX, the time-variant phase-Doppler vector $\mathbf{d}_{\text{r}}^{q,t}\in\mathbb{C}^{1\times M_{\text{r}}}$ is given by
\begin{align}
\mathbf{d}_{\text{r}}^{q,t}\left(\mathbf{v}_{\mathrm{e}}^{t},\boldsymbol{\psi}_{\text{e}}^{t}\right)=\left[e^{-\mathrm{j}2\pi f^{q,1,t}_{\text{r}}\Delta T},\ldots,e^{-\mathrm{j}2\pi f^{q,M_{\text{r}},t}_{\text{r}}\Delta T}\right],
\end{align}
where $f^{q,m,t}_{\text{r}}=\frac{v^{q,m,t}_{\text{r}}}{\lambda}$ denotes the Doppler frequency shift over the $m$-th slot on the $q$-th LCX in the $t$-th CPI. Additionally, the wireless channel vector of slots on the $q$-th LCX at the beginning of the $t$-th CPI is given by
\begin{align}
    & \boldsymbol{\alpha}_{\text{e,r}}^{q,t}\left(\boldsymbol{\psi}_{\text{e}}^{t}\right)= \left[\frac{\eta e^{-j\frac{2\pi}{\lambda}\left\|\boldsymbol{\psi}_{\text{e}}^{t}-\boldsymbol{\psi}^{q,1}_{\text{c}}\right\|}}{\left\|\boldsymbol{\psi}_{\text{e}}^{t}-\boldsymbol{\psi}^{q,1}_{\text{c}}\right\|}, ..., \frac{\eta e^{-j\frac{2\pi}{\lambda}\left\|\boldsymbol{\psi}_{\text{e}}^{t}-\boldsymbol{\psi}^{q,M_{\text{r}}}_{\text{c}}\right\|}}{\left\|\boldsymbol{\psi}_{\text{e}}^{t}-\boldsymbol{\psi}^{q,M_{\text{r}}}_{\text{c}}\right\|}\right].
    \label{eq:single-h-channel-Eve-eco}
\end{align}
Similar to Eq.~\eqref{eq:multi-overall-channel-pro}, the wireless channel vector from Eve to the $q$-th LCX is given by 
\begin{align}
\label{eq:multi-overall-channel-eco}
\mathbf{h}_{\text{e,r}}^{q,t}\left(\mathbf{v}_{\mathrm{e}}^{t},\boldsymbol{\psi}_{\text{e}}^{t}\right)\triangleq\boldsymbol{\alpha}_{\text{e,r}}^{q,t}\left(\boldsymbol{\psi}_{\text{e}}^{t}\right)\odot\mathbf{d}_{\text{r}}^{q,t}\left(\mathbf{v}_{\mathrm{e}}^{t},\boldsymbol{\psi}_{\text{e}}^{t}\right).    
\end{align}
Thus, the received echo signal on the $q$-th LCX in the $t$-th CPI is given by
\begin{align}
y_{\mathrm{c}}^{q,t}=\sqrt{\beta}\left(\mathbf{v}^{q}\right)^{T}\left(\mathbf{h}_{\mathrm{e,r}}^{q,t}\left(\mathbf{v}_{\mathrm{e}}^{t},\boldsymbol{\psi}_{\text{e}}^{t}\right)\right)^{T}y_{\mathrm{e}}^{t}+n_{\mathrm{c}}^{q,t},
\end{align}
where $n_{\mathrm{c}}^{q,t}\sim\mathcal{CN}(0,M_{\text{r}}\sigma_{\mathrm{c}}^{2})$ denotes the AWGN at the $q$-th LCX with $\sigma_{\mathrm{c}}^{2}$ denoting the noise power, $\beta\in\mathbb{C}^{1\times1}$ denotes the radar cross section (RCS) of Eve, and the combination vector $\mathbf{v}^{q}\in\mathbb{C}^{M_{\text{r}}\times1}$ is given by
\begin{align}
\mathbf{v}^{q}\triangleq\left[e^{-j\frac{2\pi n_{\text{r}}}{\lambda}x_{\text{c}}^{q,1}},...,e^{-j\frac{2\pi n_{\text{r}}}{\lambda}x_{\text{c}}^{q,M_{\text{r}}}}\right]^{T},    
\end{align}
where $n_{\text{r}}$ denotes the effective refractive index of the LCX. Jointly considering all LCXs and assuming that $n_{\mathrm{c}}^{q,t}$ are independently and identically distributed, the received echo signal can be expressed as
\begin{align}
\label{eq:received-echo-signal-multi}
&\mathbf{y}_{\mathrm{c}}^{t}=[y_{\mathrm{c}}^{1,t},...,y_{\mathrm{c}}^{N_{\text{r}},t}]^{T}=\sqrt{\beta}\mathbf{V}^{T}\left(\mathbf{h}_{\text{e,r}}^{t}\left(\mathbf{v}_{\mathrm{e}}^{t},\boldsymbol{\psi}_{\text{e}}^{t}\right)\right)^{T}y_{\mathrm{e}}^{t}+\mathbf{n}_{\mathrm{c}}^{t}\nonumber\\
&=\sqrt{\beta}\mathbf{V}^{T}\left(\mathbf{h}_{\text{e,r}}^{t}\left(\mathbf{v}_{\mathrm{e}}^{t},\boldsymbol{\psi}_{\text{e}}^{t}\right)\right)^{T}\mathbf{h}_{\text{e,t}}^{t}\left(\mathbf{v}_{\mathrm{e}}^{t},\boldsymbol{\psi}_{\text{e}}^{t},\mathbf{X}_{\text{p}}\right)\mathbf{G}\left(\mathbf{X}_{\text{p}}\right)\mathbf{c}^{t}+\mathbf{n}_{\mathrm{c}}^{t},
\end{align}
where $\mathbf{n}_\mathrm{c}^{t}\sim\mathcal{CN}(\mathbf{0}_{N_{\text{r}}},M_{\text{r}}\sigma_\mathrm{c}^2\mathbf{I}_{N_{\text{r}}})$ denotes the Gaussian noise, the combination matrix $\mathbf{V}\in\mathbb{C}^{N_{\text{r}}M_{\text{r}}\times N_{\text{r}}}$ is given by
\begin{align}
\mathbf{V}=
\begin{bmatrix}
\mathbf{v}^1 & \mathbf{0} & \ldots & \mathbf{0} \\
\mathbf{0} & \mathbf{v}^2 & \ldots & \mathbf{0} \\
\vdots & \vdots & \ddots & \vdots \\
\mathbf{0} & \mathbf{0} & \ldots & \mathbf{v}^{N_{\text{r}}}
\end{bmatrix},
\end{align}
and $\mathbf{h}_{\text{e,r}}^{t}\left(\mathbf{v}_{\mathrm{e}}^{t},\boldsymbol{\psi}_{\text{e}}^{t}\right)=\left[\mathbf{h}_{\text{e,r}}^{1,t}\left(\mathbf{v}_{\mathrm{e}}^{t},\boldsymbol{\psi}_{\text{e}}^{t}\right),\ldots,\mathbf{h}_{\text{e,r}}^{N_{\text{r}},t}\left(\mathbf{v}_{\mathrm{e}}^{t},\boldsymbol{\psi}_{\text{e}}^{t}\right)\right]\in\mathbb{C}^{1\times N_{\text{r}}M_{\text{r}}}$.

\subsection{Communication Signal Model}
The transmission to Bob via the $n$-th waveguide can be decomposed into two components: in-waveguide propagation from the feed point to PAs, and free-space propagation from PAs to Bob, denoted by $\mathbf{g}\left(\mathbf{x}_{\text{p}}^{n}\right)$ and $\mathbf{h}_{\text{b}}^n\left(\mathbf{x}_{\text{p}}^{n}\right)$, respectively. Accordingly, the received signal at Bob in the $t$-th CPI can be given by
\begin{align}
\label{eq:signal-Bob}
y_\text{b}^{t}=\mathbf{h}_{\text{b}}\left(\mathbf{X}_{\text{p}}\right)\mathbf{G}\left(\mathbf{X}_{\text{p}}\right)\mathbf{c}^{t}+n_\mathrm{b},
\end{align}
where $n_{\text{b}}\sim\mathcal{CN}(0,\sigma_{\text{b}}^{2})$ represents the AWGN at Bob, with $\sigma_{\text{b}}^{2}$ denoting the noise power, $\mathbf{h}_{\text{b}}\left(\mathbf{X}_{\text{p}}\right)=\left[\mathbf{h}_{\text{b}}^1\left(\mathbf{x}_{\mathrm{p}}^{1}\right),\ldots,\mathbf{h}_{\text{b}}^{N_{\text{t}}}\left(\mathbf{x}_{\mathrm{p}}^{N_{\text{t}}}\right)\right]\in\mathbb{C}^{1\times N_{\text{t}}M_{\text{t}}}$ denotes the channel vector from PAs over all waveguides to Bob, and $\mathbf{h}_{\text{b}}^{n}\left(\mathbf{x}_{\mathrm{p}}^{n}\right)$ is given by
\begin{align}
    & \mathbf{h}_{\text{b}}^{n}\left(\mathbf{x}_{\mathrm{p}}^{n}\right)= \left[\frac{\eta e^{-j\frac{2\pi}{\lambda}\left\|\boldsymbol{\psi}_{\text{b}}-\boldsymbol{\psi}^{n,1}_{\text{p}}\right\|}}{\left\|\boldsymbol{\psi}_{\text{b}}-\boldsymbol{\psi}^{n,1}_{\text{p}}\right\|}, ..., \frac{\eta e^{-j\frac{2\pi}{\lambda}\left\|\boldsymbol{\psi}_{\text{b}}-\boldsymbol{\psi}^{n,M_{\text{t}}}_{\text{p}}\right\|}}{\left\|\boldsymbol{\psi}_{\text{b}}-\boldsymbol{\psi}^{n,M_{\text{t}}}_{\text{p}}\right\|}\right].
\end{align}
Accordingly, the achievable data rate at Bob is given by
\begin{align}
\label{eq:rate-Bob}
R_{\text{b}}^{t}\left(\mathbf{w}^{t},\mathbf{z}^{t},\mathbf{X}_{\text{p}}\right)=\log_{2}\left(1+\frac{\left|\mathbf{h}_{\text{b}}\left(\mathbf{X}_{\text{p}}\right)\mathbf{G}\left(\mathbf{X}_{\text{p}}\right)\mathbf{w}^{t}\right|^2}{\left|\mathbf{h}_{\text{b}}\left(\mathbf{X}_{\text{p}}\right)\mathbf{G}\left(\mathbf{X}_{\text{p}}\right)\mathbf{z}^{t}\right|^2+\sigma_\mathrm{b}^2}\right).
\end{align}
Similarly, the leakage information rate at Eve is given by
\begin{align}
\label{eq:rate-Eve}
R_{\text{e}}^{t}\left(\mathbf{w}^{t},\mathbf{z}^{t},\mathbf{X}_{\text{p}}\right)=\log_{2}\left(1+\frac{\left|\mathbf{h}_{\text{e,t}}^{t}\left(\mathbf{X}_{\text{p}}\right)\mathbf{G}\left(\mathbf{X}_{\text{p}}\right)\mathbf{w}^{t}\right|^2}{\left|\mathbf{h}_{\text{e,t}}^{t}\left(\mathbf{X}_{\text{p}}\right)\mathbf{G}\left(\mathbf{X}_{\text{p}}\right)\mathbf{z}^{t}\right|^2+\sigma_\mathrm{e}^2}\right).
\end{align}
Based on~\eqref{eq:rate-Bob} and~\eqref{eq:rate-Eve}, the secrecy rate in the $t$-th CPI can
be expressed as follows:
\begin{align}
R_{\text{s}}^{t}=\left[R_{\text{b}}^{t}-R_{\text{e}}^{t}\right]^+.
\end{align}

\section{EKF-Based Eve Tracking Scheme and Problem Formulation for Multiple-Waveguide PASS}
To obtain the Eve’s CSI, an EKF-based state tracking method is proposed to track the Eve's trajectory. Then, we formulate the average secrecy rate maximization problem.

\subsection{EKF Based Eve State Tracking}
\label{subsec:EKF}
To obtain the CSI of Eve, we propose an EKF-based method for tracking Eve’s movement, i.e., locations and velocities. In the $t$-th CPI, denote the velocity variances along the $x$- and $y$-axes with Gaussian random variables $\Delta v_{x}$ and $\Delta v_{y}$, respectively, each with zero mean and variances $\sigma_{v_x}^2$ and $\sigma_{v_y}^2$, respectively. Then, we introduce the kinematic model, which characterizes the evolution of Eve’s mobility state. Specifically, the model can be expressed as
\begin{subequations}
\label{eq:kinematic-model}
\begin{equation}
\label{eq:kinematic-model-vx}
v_{x}^{t}=v_{x}^{t-1}+\Delta v_{x},
\end{equation}
\begin{equation}
\label{eq:kinematic-model-vy}
v_{y}^{t}=v_{y}^{t-1}+\Delta v_{y},
\end{equation}
\begin{equation}
x_{\mathrm{e}}^{t}=x_{\mathrm{e}}^{t-1}+v_{x}^{t-1}\Delta T,
\end{equation}
\begin{equation}
y_{\mathrm{e}}^{t}=y_{\mathrm{e}}^{t-1}+v_{y}^{t-1}\Delta T.
\end{equation}
\end{subequations}
In the following, we first construct the state
transition and observation models, and then introduce the EKF framework, which consists of two main steps: prior prediction and posterior update.

\subsubsection{State Transition and Observation Model}
Denote the mobility status vector of Eve under the kinematic model as $\boldsymbol{\xi}^{t}=[x_{\mathrm{e}}^{t},y_{\mathrm{e}}^{t},v_{x}^{t},v_{y}^{t}]^{T}\in\mathbb{R}^{4\times1}$. Therefore, based on Eq.~\eqref{eq:kinematic-model}, the state transition model from the $(t-1)$-th CPI to the $t$-th CPI is given by
\begin{align}
\label{eq:state-transition-model}
\boldsymbol{\xi}^{t}=g(\boldsymbol{\xi}^{t-1})+\boldsymbol{\omega}^{t},
\end{align}
where $g(\cdot)$ denotes the states transition matrix, given by
\begin{align}
\label{eq:states-transition-matrix}
g\left(\boldsymbol{\xi}^{t-1}\right)=
\begin{bmatrix}
1 & 0 & \Delta T & 0 \\
0 & 1 & 0 & \Delta T \\
0 & 0 & 1 & 0 \\
0 & 0 & 0 & 1
\end{bmatrix}\boldsymbol{\xi}^{t-1}.
\end{align}
Still referring to Eq.~\eqref{eq:state-transition-model}, $\boldsymbol{\omega}^{t}=\left[0,0,\Delta v_{x},\Delta v_{y}\right]^{T}\in\mathbb{R}^{4\times1}$. Since $\Delta v_{x}$ and $\Delta v_{y}$ are assumed to be mutually independent, we have $\boldsymbol{\omega}^{t}\sim\mathcal{N}(\mathbf{0}_4,\mathbf{Q}_s)$ with $\mathbf{Q}_{s}=\mathrm{diag}\{0,0,\sigma_{v_x}^2,\sigma_{v_y}^2\}\in\mathbb{R}^{4\times4}$.

The observation model establishes the relationship between the received echo signals at the BS and Eve's mobility state. We define the observation model of Eve at the $t$-th CPI as
\begin{align}
\mathbf{y}_\text{c}^{t}\triangleq h\left(\boldsymbol{\xi}^{t}\right)+\mathbf{n}_{\mathrm{c}}^{t},
\end{align}
where $h\left(\cdot\right)$ denotes the function of the echo signal with respect to $\boldsymbol{\xi}^{t}$ according to Eq.~\eqref{eq:received-echo-signal-multi}.

\subsubsection{Prior Prediction}
To accurately characterize the temporal motion of Eve, it is essential to predict the next mobility status $\boldsymbol{\xi}^{t}$ based on the current mobility
state $\boldsymbol{\xi}^{t-1}$. According to Eq.~\eqref{eq:state-transition-model}, the current mobility status is fed to the states transition matrix $g(\cdot)$ to perform the prior prediction on the mobility status:
\begin{align}
\label{eq:prior-prediction}
\hat{\boldsymbol{\xi}}^{t}=g(\boldsymbol{\xi}^{t-1}).
\end{align}
The corresponding prediction error covariance
matrix $\mathbf{P}^{t|t-1}$ can be updated according to~\cite{9171304}
\begin{align}
\mathbf{P}^{t|t-1}=\mathbf{G}^{t-1}\mathbf{P}^{t-1}\left(\mathbf{G}^{t-1}\right)^{H}+\mathbf{Q}_{s},
\end{align}
where $\mathbf{G}^{t-1}$ denotes the Jacobian matrix of the kinematic model in Eq.~\eqref{eq:states-transition-matrix}, and $\mathbf{P}^{t-1}$ denotes the error covariance matrix of the mobility state at the $(t-1)$-th CPI.

\subsubsection{Posterior Update}
After receiving the echo signal in the $t$-th CPI, the prior estimate of Eve’s mobility state is updated accordingly. With the predicted mobility state $\hat{\boldsymbol{\xi}}^{t}$ and the received echo signal $\mathbf{y}_{\mathrm{c}}^{t}$, the Eve’s mobility state vector in the $t$-th CPI is updated by 
\begin{align}
\label{eq:posterior-update}
\tilde{\boldsymbol{\xi}}^{t}=\hat{\boldsymbol{\xi}}^{t}+\mathbf{K}^{t}\left(\mathbf{y}_{\mathrm{c}}^{t}-h\left(\hat{\boldsymbol{\xi}}^{t}\right)\right),
\end{align}
where $\mathbf{K}^{t}$ denotes the Kalman gain that is given by
\begin{align}
\label{eq:Kalman-gain}
\mathbf{K}^{t}=\mathbf{P}^{t|t-1}\left(\mathbf{J}^{t|t-1}\right)^H\left(\mathbf{S}^{t}\right)^{-1},
\end{align}
with $\mathbf{J}^{t|t-1}=\frac{\partial h\left(\boldsymbol{\xi}\right)}{\partial\boldsymbol{\xi}}\mid_{\boldsymbol{\xi}=\hat{\boldsymbol{\xi}}^{t}}$ denotes the Jacobian matrix of the observation model with respect to the mobility state. Still referring to Eq.~\eqref{eq:Kalman-gain}, $\mathbf{S}^{t}$ is the residual covariance matrix, given by
\begin{align}
\mathbf{S}^{t}=\mathbf{J}^{t|t-1}\mathbf{P}^{t|t-1}\left(\mathbf{J}^{t|t-1}\right)^H+\mathbf{R},
\end{align}
where $\mathbf{R}=M_{\text{r}}\sigma_\mathrm{c}^2\mathbf{I}_{N_\mathrm{r}}$. Finally, the error covariance matrix of the mobility state at the $t$-th CPI is updated according to
\begin{align}
\mathbf{P}^{t}=
\begin{pmatrix}
\mathbf{I}_4-\mathbf{K}^{t}\mathbf{J}^{t|t-1}
\end{pmatrix}\mathbf{P}^{t|t-1}.
\end{align}

\subsection{Problem Formulation}
With the proposed Eve tracking scheme, the BS can obtain Eve’s locations and velocities, and thus facilitates the subsequent pinching and baseband beamforming designs. On this basis, and given the similarity of the problem across each time block, we formulate the optimization problem with the aim of maximizing the average secrecy rate in one time block. Specifically, the optimization problem is given by
\begin{subequations}
\label{eq:optimization-problem-multiple}
\begin{equation}
\label{eq:objective-function-multiple}
\max_{\{\mathbf{w}^{t},\mathbf{z}^{t}\},\mathbf{X}_{\text{p}}}\frac{1}{T}\sum_{t=1}^{T} R_{\text{s}}^{t},
\end{equation}
\begin{equation}
\label{eq:maximum-power-multiple}
{\rm{s.t.}} \ \
\|\mathbf{w}^{t}\|^2+\|\mathbf{z}^{t}\|^2\leq P_{\max},\forall t,
\end{equation}
\begin{equation}
\label{eq:echo-signal-power-multiple}
p(\mathbf{w}^{t},\mathbf{z}^{t},\mathbf{X}_{\text{p}})\geq\Gamma,\forall t,
\end{equation}
\begin{equation}
\label{eq:position_constraint1-multiple}
0\leq x_{\text{p}}^{n,m}\leq L,\forall n,m,
\end{equation}
\begin{equation}
\label{eq:position_constraint2-multiple}
x_{\mathrm{p}}^{n,m+1}-x_{\mathrm{p}}^{n,m}\geq\Delta,\forall n,m.
\end{equation}
\end{subequations}
Constraint~\eqref{eq:maximum-power-multiple} restricts that the BS transmit power does not exceed the maximum budget $P_{\max}$. In constraint~\eqref{eq:echo-signal-power-multiple}, $p(\mathbf{c}^t,\mathbf{X}_{\text{p}})$ denotes the power of the received echo signal, which is given by
\begin{align}
\begin{aligned}
p(\mathbf{c}^t,\mathbf{X}_{\text{p}})\triangleq
\left(\mathbf{c}^{t}\right)^{H}\left(\mathbf{h}_{\mathrm{e}}^{t}\left(\mathbf{v}_{\mathrm{e}}^t,\boldsymbol{\psi}_{\text{e}}^t,\mathbf{X}_{\text{p}}\right)\right)^{H}\mathbf{h}_{\mathrm{e}}^{t}\left(\mathbf{v}_{\mathrm{e}}^t,\boldsymbol{\psi}_{\text{e}}^t,\mathbf{X}_{\text{p}}\right)\mathbf{c}^{t},
\end{aligned}
\end{align}
where $\mathbf{h}_{\text{e}}^{t}\left(\mathbf{v}_{\mathrm{e}}^{t},\boldsymbol{\psi}_{\text{e}}^{t},\mathbf{X}_{\text{p}}\right)$ denotes the round-trip channel vector in the $t$-th CPI, given by
\begin{align}
&\mathbf{h}_{\text{e}}^{t}\left(\mathbf{v}_{\mathrm{e}}^{t},\boldsymbol{\psi}_{\text{e}}^{t},\mathbf{X}_{\text{p}}\right)\nonumber\\&\triangleq\sqrt{\beta}\mathbf{V}^{T}\left(\mathbf{h}_{\text{e,r}}^{t}\left(\mathbf{v}_{\mathrm{e}}^{t},\boldsymbol{\psi}_{\text{e}}^{t}\right)\right)^{T}\mathbf{h}_{\text{e,t}}^{t}\left(\mathbf{v}_{\mathrm{e}}^{t},\boldsymbol{\psi}_{\text{e}}^{t},\mathbf{X}_{\text{p}}\right)\mathbf{G}\left(\mathbf{X}_{\text{p}}\right).
\end{align}
Constraint~\eqref{eq:echo-signal-power-multiple} ensures that the power of the received echo signal exceeds a predefined threshold $\Gamma$. Moreover, constraint~\eqref{eq:position_constraint1-multiple}
restricts the length of waveguides. Constraint~\eqref{eq:position_constraint2-multiple} guarantees
the minimum spacing among PAs to be no smaller than $\Delta$ for avoiding antenna coupling. Note that, due to the inherent nature of the two-timescale beamforming framework, on the one hand, the expression of $R_{\text{s}}^{t},\forall t$, for the optimization of pinching beamforming $\mathbf{X}_{\text{p}}$ is based on the prior knowledge of Eve's locations obtained by the prediction model given in Eq.~\eqref{eq:prior-prediction}. On the other hand, since the baseband beamforming $\left\{\mathbf{w}^{t},\mathbf{z}^{t}\right\}$ can be updated across each CPI, $R_{\text{s}}^{t},\forall t$ can be given in a more precise manner for the baseband beamforming optimization, by refining the Eve's positions with the echo signals based on Eq.~\eqref{eq:posterior-update}. To address this challenging issue of optimization across different timescales, we propose a corresponding two-timescale optimization approach, which is detailed in the following section.

\section{Proposed Two-Timescale Solution for Multiple-Waveguide PASS}
Although the baseband beamforming $\left\{\mathbf{w}^{t},\mathbf{z}^{t}\right\}$ can be updated across each CPI, it is coupled with the pinching beamforming as observed in~\eqref{eq:optimization-problem-multiple}. Therefore, for effectively solving the two-timescale beamforming problem~\eqref{eq:optimization-problem-multiple}, we propose to jointly optimize $\mathbf{X}_{\text{p}}$ and $\left\{\mathbf{w}^{t},\mathbf{z}^{t}\right\}$ at the beginning of each time block based on the predicted channels, and continuously update $\left\{\mathbf{w}^{t},\mathbf{z}^{t}\right\}$  with the refined channels during each CPI, which facilitates the two-timescale optimization algorithm. We invoke the AO algorithm for the joint optimization of $\mathbf{X}_{\text{p}}$ and $\left\{\mathbf{w}^{t},\mathbf{z}^{t}\right\}$, where the pinching and baseband beamforming are solved in an iterative manner. In the following, we first introduce the optimization methods for optimizing $\mathbf{X}_{\text{p}}$ and $\left\{\mathbf{w}^{t},\mathbf{z}^{t}\right\}$ individually, and then present the overall two-timescale optimization algorithm.

\subsection{Large-timescale Pinching Beamforming Optimization}
With fixed baseband beamforming vectors $\{\mathbf{w}^{t},\mathbf{z}^{t}\}$, the subproblem with
respect to $\mathbf{X}_{\text{p}}$ can be formulated as
\begin{subequations}
\label{eq:optimization-problem-multiple-X-1}
\begin{equation}
\label{eq:objective-function-multiple-X-1}
\max_{\mathbf{X}_{\text{p}}}\frac{1}{T}\sum_{t=1}^{T} R_{\text{s}}^{t},
\end{equation}
\begin{equation}
\label{eq:echo-signal-power-multiple-X-1}
{\rm{s.t.}} \ \
p\left(\mathbf{X}_{\text{p}}\right)\geq\Gamma,\forall t,
\end{equation}
\begin{equation}
\label{eq:position_constraint1-multiple-X-1}
0\leq x_{\text{p}}^{n,m}\leq L,\forall n,m,
\end{equation}
\begin{equation}
\label{eq:position_constraint2-multiple-X-1}
x_{\text{p}}^{n,m+1}-x_{\text{p}}^{n,m}\geq\Delta,\forall n,m.
\end{equation}
\end{subequations}
It can be observed that elements $\left\{x_{\text{p}}^{n,m}\right\},\forall n,m$ of $\mathbf{X}_{\text{p}}$ are coupled in the objective function and constraints~\eqref{eq:echo-signal-power-multiple-X-1},~\eqref{eq:position_constraint2-multiple-X-1}, which makes~\eqref{eq:optimization-problem-multiple-X-1} an intractable problem. To tackle this issue, we propose an element-wise alternating optimization method, in which each $x_{\text{p}}^{n,m}\in\mathbf{X}_{\text{p}}$ is optimized while keeping the other elements fixed.

\begin{figure*}
\begin{subequations}
\label{eq:optimization-problem-multiple-X-2}
\begin{equation}
\label{eq:objective-function-multiple-X-2}
\max_{x_{\text{p}}^{n,m}}\frac{1}{T}\sum_{t=1}^{T} \left[\log_{2}\left(1+\frac{\left|A_{\text{b}}^{t}\left(x_{\text{p}}^{n,m}\right)\right|^2}{\left|C_{\text{b}}^{t}\left(x_{\text{p}}^{n,m}\right)\right|^2+M_{\text{t}}\sigma_\mathrm{b}^2}\right)-\log_{2}\left(1+\frac{\left|A_{\text{e}}^{t}\left(x_{\text{p}}^{n,m}\right)\right|^2}{\left|C_{\text{e}}^{t}\left(x_{\text{p}}^{n,m}\right)\right|^2+M_{\text{t}}\sigma_\mathrm{e}^2}\right)\right]^+,
\end{equation}
\begin{equation}
\label{eq:echo-signal-power-multiple-X-2}
{\rm{s.t.}} \ \
g\left(x_{\text{p}}^{n,m}\right)\geq\Gamma,\forall t,
\end{equation}
\begin{equation}
\label{eq:position_constraint-multiple-X-2}
\eqref{eq:position_constraint1-multiple-X-1},\eqref{eq:position_constraint2-multiple-X-1}.
\notag
\end{equation}
\end{subequations}
\hrulefill
\end{figure*}
Specifically, the subproblem with respect to $x_{\text{p}}^{n,m}$ is given by~\eqref{eq:optimization-problem-multiple-X-2}, shown at the top of the this page. $A_{\text{b}}\left(x_{\text{p}}^{n,m}\right)$, $C_{\text{b}}\left(x_{\text{p}}^{n,m}\right)$, $A_{\text{e}}^t\left(x_{\text{p}}^{n,m}\right)$ and $C_{\text{e}}^t\left(x_{\text{p}}^{n,m}\right)$ represent the effective signal terms that depends on $x_{\text{p}}^{n,m}$, given by
\begin{subequations}
\begin{equation}
A_{\text{b}}^{t}\left(x_{\text{p}}^{n,m}\right)=w_{n}^{t}\frac{\eta e^{-j\theta^{n,m}}}{\left\|\psi_{\text{b}}-\psi_{\mathrm{p}}^{n,m}\right\|}+\sum_{i=1}^{N_{\text{t}}}\sum_{q\neq m}^{M_{\text{t}}}w_{i}^{t}\frac{\eta e^{-j\theta^{i,q}}}{\left\|\psi_{\text{b}}-\psi_{\mathrm{p}}^{i,q}\right\|},
\end{equation}
\begin{equation}
C_{\text{b}}^{t}\left(x_{\text{p}}^{n,m}\right)=z_{n}^{t}\frac{\eta e^{-j\theta^{n,m}}}{\left\|\psi_{\text{b}}-\psi_{\mathrm{p}}^{n,m}\right\|}+\sum_{i=1}^{N_{\text{t}}}\sum_{q\neq m}^{M_{\text{t}}}z_{i}^{t}\frac{\eta e^{-j\theta^{i,q}}}{\left\|\psi_{\text{b}}-\psi_{\mathrm{p}}^{i,q}\right\|},
\end{equation}
\begin{equation}
A_{\text{e}}^t\left(x_{\text{p}}^{n,m}\right)=w_{n}^{t}\frac{\eta e^{-j\phi^{t,n,m}}}{\left\|\psi_{\text{e}}^t-\psi_{\mathrm{p}}^{n,m}\right\|}+\sum_{i=1}^{N_{\text{t}}}\sum_{q\neq m}^{M_{\text{t}}}w_{i}^{t}\frac{\eta e^{-j\phi^{t,i,q}}}{\left\|\psi_{\text{e}}^t-\psi_{\mathrm{p}}^{i,q}\right\|},
\end{equation}
\begin{equation}
C_{\text{e}}^t\left(x_{\text{p}}^{n,m}\right)=z_{n}^{t}\frac{\eta e^{-j\phi^{t,n,m}}}{\left\|\psi_{\text{e}}^t-\psi_{\mathrm{p}}^{n,m}\right\|}+\sum_{i=1}^{N_{\text{t}}}\sum_{q\neq m}^{M_{\text{t}}}z_{i}^{t}\frac{\eta e^{-j\phi^{t,i,q}}}{\left\|\psi_{\text{e}}^t-\psi_{\mathrm{p}}^{i,q}\right\|},
\end{equation}
\end{subequations}
where $w_{n}^{t}$ and $z_{n}^{t}$ represent the $n$-th element of $\mathbf{w}^{t}$ and $\mathbf{z}^{t}$, respectively. Moreover, $\theta^{n,m}=\frac{2\pi}{\lambda}\left\|\boldsymbol{\psi}_{\text{b}}-\boldsymbol{\psi}^{n,m}_{\text{p}}\right\|+\frac{2\pi }{\lambda_{\text{g}}}x_{\text{p}}^{n,m}$ and $\phi^{t,n,m}=\frac{2\pi}{\lambda}\left(\left\|\boldsymbol{\psi}_{\text{e}}^{t}-\boldsymbol{\psi}^{n,m}_{\text{p}}\right\|+\Delta Tv^{t,n,m}\right)+\frac{2\pi }{\lambda_{\text{g}}}x_{\text{p}}^{n,m}$ denote the phase-shifts due to wave propagation both inside the waveguide and in the free space. 
Observing that~\eqref{eq:optimization-problem-multiple-X-2} is a single-variable optimization problem with restricted feasible set, which can be effectively solved with the one-dimensional search, so as to avoid local optimum. Specifically, the feasible set in~\eqref{eq:position_constraint-multiple-X-2} is discretized into a uniformly spaced grid of candidate locations $\mathcal{S}_{x}$ to enable fine search, which is given by
\begin{align}
\label{eq:feasible-set}
\mathcal{S}_x=\left\{0,\frac{L}{N-1},\frac{2L}{N-1},\ldots,L\right\},
\end{align}
where $N=|\mathcal{S}_{x}|$ denotes the number of searching grids. Therefore, the optimal solution of $x_{\text{p}}^{n,m}$ is searched over the fine grid subject to the constraints~\eqref{eq:echo-signal-power-multiple-X-2} and~\eqref{eq:position_constraint-multiple-X-2}. The detailed algorithm for solving the pinching beamforming problem~\eqref{eq:optimization-problem-multiple-X-1} is given in~\textbf{Algorithm~\ref{alg:element-wise}}.

\begin{algorithm}[tp]
	\caption{Proposed Element-wise Optimization Algorithm for Solving Problem~\eqref{eq:optimization-problem-multiple-X-1}}
        \label{alg:element-wise}
    \LinesNumbered
    \KwIn{Bob's location $\boldsymbol{\psi}_{\text{b}}$, predicted Eve's mobility states $\{\hat{\boldsymbol{\xi}}^{t}\}_{t=1}^{T}$, convergence criterion $\epsilon$}
    Initialize the pinching beamforming $\mathbf{X}_{\text{p}}$;\\
    \Repeat{the increment of $\frac{1}{T}\sum_{t=1}^{T} R_{\text{s}}^{t}$ is below $\epsilon$}
    {
    \For{$n=1$ \KwTo $N_{\text{t}}$}
    {
    \For{$m=1$ \KwTo $M_{\text{t}}$}
    {
    Update $x_{\text{p}}^{n,m}$ by solving problem~\eqref{eq:optimization-problem-multiple-X-2} through one-dimensional search\;
    }
    }
    }
    \KwOut{Optimized pinching beamforming $\mathbf{X}_{\text{p}}$.}
\end{algorithm}

\subsection{Small-timescale Baseband Beamforming Optimization}
\label{subsec:baseband-beamforming}
Assuming fixed PAs positions $\mathbf{X}_{\text{p}}$ and introducing $\mathbf{W}^{t}=\mathbf{w}^{t}\left(\mathbf{w}^{t}\right)^{H}\in\mathbb{C}^{N_{\text{t}}\times N_{\text{t}}}$, $\mathbf{Z}^{t}=\mathbf{z}^{t}\left(\mathbf{z}^{t}\right)^{H}\in\mathbb{C}^{N_{\text{t}}\times N_{\text{t}}}$ and $\mathbf{H}_\text{e}^{t}=\left(\mathbf{h}_\text{e}^{t}\right)^H\mathbf{h}_\text{e}^{t}\in\mathbb{C}^{N_{\text{t}}\times N_{\text{t}}}$, problem~\eqref{eq:optimization-problem-multiple} degenerates into the following baseband beamforming problem:
\begin{subequations}
\label{eq:optimization-problem-multiple-wz}
\begin{equation}
\label{eq:objective-function-multiple-wz}
\max_{\{\mathbf{W}^{t},\mathbf{Z}^{t}\}} R_{\text{s}}^{t},
\end{equation}
\begin{equation}
\label{eq:maximum-power-multiple-wz}
{\rm{s.t.}} \ \
\mathrm{Tr}\left(\mathbf{W}^{t}\right)+\mathrm{Tr}\left(\mathbf{Z}^{t}\right)\leq P_{\max},\forall t,
\end{equation}
\begin{equation}
\label{eq:echo-signal-power-multiple-wz}
\mathrm{Tr}\left(\mathbf{H}_\text{e}^{t}\mathbf{W}^{t}\right)+\mathrm{Tr}\left(\mathbf{H}_\text{e}^{t}\mathbf{Z}^{t}\right)\geq\Gamma,\forall t,
\end{equation}
\begin{equation}
\label{eq:semi-positive-constraint}
\mathbf{W}^{t},\mathbf{Z}^{t}\succeq0,\forall t,
\end{equation}
\begin{equation}
\label{eq:rank1-constraint-W}
\mathrm{Rank}\left(\mathbf{W}^{t}\right)=1,\forall t,
\end{equation}
\begin{equation}
\label{eq:rank1-constraint-Z}
\mathrm{Rank}\left(\mathbf{Z}^{t}\right)=1,\forall t.
\end{equation}
\end{subequations}
For the ease of tractability, we further rewrite $R_{\text{s}}^{t}$ as follows:
\begin{align}
R_{\text{s}}^{t}&=\left[\log_2\left(1+\frac{\mathrm{Tr}\left(\bar{\mathbf{h}}_\text{b}^H\mathbf{W}^{t}\bar{\mathbf{h}}_\text{b}\right)}{\mathrm{Tr}\left(\bar{\mathbf{h}}_\text{b}^H\mathbf{Z}^{t}\bar{\mathbf{h}}_\text{b}\right)+\sigma_\text{b}^2}\right)\right.\nonumber\\
&-\left.\log_2\left(1+\frac{\mathrm{Tr}\left(\left(\bar{\mathbf{h}}_\text{e,t}^{t}\right)^H\mathbf{W}^{t}\bar{\mathbf{h}}_\text{e,t}^{t}\right)}{\mathrm{Tr}\left(\left(\bar{\mathbf{h}}_\text{e,t}^{t}\right)^H\mathbf{Z}^{t}\bar{\mathbf{h}}_\text{e,t}^{t}\right)+\sigma_\text{e}^2}\right)\right]^+,
\end{align}
where $\bar{\mathbf{h}}_\text{b}=\mathbf{h}_{\text{b}}\left(\mathbf{X}_{\text{p}}\right)\mathbf{G}\left(\mathbf{X}_{\text{p}}\right)$ and $\bar{\mathbf{h}}_\text{e,t}^{t}=\mathbf{h}_{\text{e,t}}^{t}\left(\mathbf{X}_{\text{p}}\right)\mathbf{G}\left(\mathbf{X}_{\text{p}}\right)$.

To solve the non-convex problem~\eqref{eq:optimization-problem-multiple-wz}, we first omit the rank-one constraints~\eqref{eq:rank1-constraint-W} and~\eqref{eq:rank1-constraint-Z}, and directly solve the relaxed version of the problem with \textbf{Lemma~\ref{lemma:rank1}}.
\begin{lemma}
\label{lemma:rank1}
The optimal transmit covariance $\mathbf{W}^{t}$ and the AN covariance $\mathbf{Z}^{t}$ of the relaxed version of problem~\eqref{eq:optimization-problem-multiple-wz} always satisfy the rank-one constraint~\eqref{eq:rank1-constraint-W} and~\eqref{eq:rank1-constraint-Z}.
\end{lemma}
\begin{proof}
The proof follows similar steps to those in~\cite[Lemma~2]{11202497}. Hence, we omit the details here.
\end{proof}
Next, we introduce the exponential auxiliary variables $\tau^{t},\varepsilon^{t},u^{t},v^{t}$ as follows:
\begin{subequations}
\begin{equation}
e^{\tau^{t}}=\mathrm{Tr}\left(\bar{\mathbf{H}}_\text{b}\mathbf{W}^{t}\right)+\mathrm{Tr}\left(\bar{\mathbf{H}}_\text{b}\mathbf{Z}^{t}\right)+\sigma_\text{b}^2,
\end{equation}
\begin{equation}
e^{\varepsilon^{t}}=\mathrm{Tr}\left(\bar{\mathbf{H}}_\text{b}\mathbf{Z}^{t}\right)+\sigma_\text{b}^2,
\end{equation}
\begin{equation}
e^{u^{t}}=\mathrm{Tr}\left(\bar{\mathbf{H}}_\text{e,t}^{t}\mathbf{W}^{t}\right)+\mathrm{Tr}\left(\bar{\mathbf{H}}_\text{e,t}^{t}\mathbf{Z}^{t}\right)+\sigma_\text{e}^2,
\end{equation}
\begin{equation}
e^{v^{t}}=\mathrm{Tr}\left(\bar{\mathbf{H}}_\text{e,t}^{t}\mathbf{Z}^{t}\right)+\sigma_\text{e}^2,
\end{equation}
\end{subequations}
where $\bar{\mathbf{H}}_\text{b}=\bar{\mathbf{h}}_\text{b}^H\bar{\mathbf{h}}_\text{b}$ and $\bar{\mathbf{H}}_\text{e,t}^{t}=\left(\bar{\mathbf{h}}_\text{e,t}^{t}\right)^H\bar{\mathbf{h}}_\text{e,t}^{t}$. Then, problem~\eqref{eq:optimization-problem-multiple-wz} can be equivalently transformed to the following one:
\begin{subequations}
\label{eq:optimization-problem-multiple-wz-2}
\begin{equation}
\label{eq:objective-function-multiple-wz-2}
\max_{\tau^{t},\varepsilon^{t},u^{t},v^{t},\mathbf{W}^{t},\mathbf{Z}^{t}} \tau^{t}-\varepsilon^{t}-u^{t}+v^{t},
\end{equation}
\begin{equation}
\label{eq:constraint-tau}
{\rm{s.t.}} \ \
e^{\tau^{t}}\leq\mathrm{Tr}\left(\bar{\mathbf{H}}_\text{b}\mathbf{W}^{t}\right)+\mathrm{Tr}\left(\bar{\mathbf{H}}_\text{b}\mathbf{Z}^{t}\right)+\sigma_\text{b}^2,\forall t,
\end{equation}
\begin{equation}
\label{eq:constraint-varepsilon}
e^{\varepsilon^{t}}\geq\mathrm{Tr}\left(\bar{\mathbf{H}}_\text{b}\mathbf{Z}^{t}\right)+\sigma_\text{b}^2,\forall t,
\end{equation}
\begin{equation}
\label{eq:constraint-u}
e^{u^{t}}\geq\mathrm{Tr}\left(\bar{\mathbf{H}}_\text{e,t}^{t}\mathbf{W}^{t}\right)+\mathrm{Tr}\left(\bar{\mathbf{H}}_\text{e,t}^{t}\mathbf{Z}^{t}\right)+\sigma_\text{e}^2,\forall t,
\end{equation}
\begin{equation}
\label{eq:constraint-v}
e^{v^{t}}\leq\mathrm{Tr}\left(\bar{\mathbf{H}}_\text{e,t}^{t}\mathbf{Z}^{t}\right)+\sigma_\text{e}^2,\forall t,
\end{equation}
\begin{equation}
\eqref{eq:maximum-power-multiple-wz}, \eqref{eq:echo-signal-power-multiple-wz}, \eqref{eq:semi-positive-constraint}.
\notag
\end{equation}
\end{subequations}
Obviously,~\eqref{eq:constraint-tau}-\eqref{eq:constraint-v} hold equality at the
optimum point, which follows from the monotonicity of the objective function. Nevertheless, constraints~\eqref{eq:constraint-varepsilon} and~\eqref{eq:constraint-u} remain non-convex. To convert the non-convex constraints into convex ones, we employ the SCA method. Specifically, the first-order Taylor expansions of~\eqref{eq:constraint-varepsilon} and~\eqref{eq:constraint-u} are given by
\begin{align}
\label{eq:SCA-varepsilon}
\mathrm{Tr}\left(\bar{\mathbf{H}}_\text{b}\mathbf{Z}^{t}\right)+\sigma_\text{b}^2\leq e^{\varepsilon^{t,(k)}}\left(\varepsilon^{t}-\varepsilon^{t,(k)}+1\right),
\end{align}
and
\begin{align}
\label{eq:SCA-U}
\mathrm{Tr}\left(\bar{\mathbf{H}}_\text{e,t}^{t}\mathbf{W}^{t}\right)+\mathrm{Tr}\left(\bar{\mathbf{H}}_\text{e,t}^{t}\mathbf{Z}^{t}\right)+\sigma_\text{e}^2\leq e^{u^{t,(k)}}\left(u^{t}-u^{t,(k)}+1\right),
\end{align}
respectively, where $\varepsilon^{t,(k)}$ and $u^{t,(k)}$ are the optimal solutions at the $k$-th iteration. By replacing constraints~\eqref{eq:constraint-varepsilon} and~\eqref{eq:constraint-u} with inequalities~\eqref{eq:SCA-varepsilon} and~\eqref{eq:SCA-U}, respectively, problem~\eqref{eq:optimization-problem-multiple-wz-2} can be transformed into a convex problem, which can be solved using convex optimization tool box such as CVX~\cite{boyd2004convex}.

\begin{algorithm}[tp]
	\caption{Proposed Two-Timescale Optimization Algorithm for Solving Problem~\eqref{eq:optimization-problem-multiple}}
        \label{alg:AO}
    \LinesNumbered
    \KwIn{Bob's location $\boldsymbol{\psi}_{\text{b}}$, convergence criterion $\varepsilon$}
    \textbf{At the beginning of each time block (Large-timescale pinching beamforming):}\\
    Predict Eve’s mobility states over the next $T$ CPIs $\{\hat{\boldsymbol{\xi}}^{t}\}_{t=1}^{T}$ with the prediction model given in Eq.~\eqref{eq:prior-prediction};\\
	Initialize the beamforming matrices $\mathbf{X}_{\text{p}}^{(0)}$ and $\{\mathbf{W}^{t,(0)},\mathbf{Z}^{t,(0)}\}_{t=1}^{T}$;\\
    Set iteration index $j=1$;\\
    \Repeat{the increment of $\frac{1}{T}\sum_{t=1}^{T} R_{\text{s}}^{t}$ is below $\varepsilon$}
    {
    Update $\mathbf{X}_{\text{p}}^{(j)}$ by \textbf{Algorithm~\ref{alg:element-wise}};\\
    Update $\{\mathbf{W}^{t,(j)},\mathbf{Z}^{t,(j)}\}_{t=1}^{T}$ by solving problem~\eqref{eq:optimization-problem-multiple-wz-2};\\
    $j=j+1$;
    }
    \KwOut{Pinching beamforming $\mathbf{X}_{\text{p}}^{*}$ for the current time block;}
    \textbf{In each CPI of the time block (Small-timescale baseband beamforming):}\\
    \For{$t=1$ \KwTo $T$}
    {
    Update Eve's mobility state $\tilde{\boldsymbol{\xi}}^{t}$ with the received echo signals based on Eq.~\eqref{eq:posterior-update}\;
    Update $\mathbf{W}^{t}$ and $\mathbf{Z}^{t}$ by solving problem~\eqref{eq:optimization-problem-multiple-wz-2}\;
    \KwOut{$\mathbf{W}^{t},\mathbf{Z}^{t}$ for the current CPI.}
    }
\end{algorithm}

\subsection{Overall Two-Timescale Optimization Algorithm Design and Property Analysis}
The proposed two-timescale optimization algorithm for solving problem~\eqref{eq:optimization-problem-multiple} is summarized in \textbf{Algorithm~\ref{alg:AO}}. Specifically, in lines 3-9, the AO method is applied at the beginning of each time block to determine the pinching beamforming. Thereafter, in lines 11-13, while the PA positions remain fixed, the baseband beamforming is updated with the refined Eves' CSI in each CPI.
\subsubsection{Convergence Analysis}
At each AO iteration, the element-wise sequential search for $\mathbf{X}_{\text{p}}$ and the SCA update for $\left\{\mathbf{w}^{t},\mathbf{z}^{t}\right\}$ both yield non-decreasing values of the original objective in~\eqref{eq:optimization-problem-multiple}. Since the sum secrecy rate is limited by the feasible set $\mathcal{S}_x$ and the total transmit power constraint, the proposed SCA and AO algorithms are guaranteed to converge.
\subsubsection{Complexity Analysis}
The complexity of the proposed two-timescale algorithm consists of two parts: the large timescale AO method carried out at the beginning of each time block, and the small timescale baseband beamforming updates performed in each CPI. For the large timescale AO method, on the one hand, the computational complexity of the element-wise pinching beamforming is given by $\mathcal{O}\left(I_{1}NTN_{\text{t}}M_{\text{t}}\right)$, where $I_{1}$ denotes the number of outer iterations in the element-wise method. On the other hand, If the interior point method is employed in the SCA for solving the baseband beamforming problem, the computational complexity is given by $\mathcal{O}\left(I_{2}N_{\text{t}}^{3.5}\right)$, where $I_{2}$ denotes the number of the SCA iterations. As a result, assuming $I_{3}$ iterations in the AO framework, the total computational complexity of the AO method is given by $\mathcal{O}\left(I_{3}\left(NTN_{\text{t}}M_{\text{t}}+I_{2}N_{\text{t}}^{3.5}\right)\right)$. For the small-timescale baseband beamforming, the corresponding complexity is given by $\mathcal{O}\left(TI_{2}N_{\text{t}}^{3.5}\right)$. As such, the overall computational complexity of \textbf{Algorithm~\ref{alg:AO}} is given by $\mathcal{O}\left(I_{3}\left(I_{1}NTN_{\text{t}}M_{\text{t}}+I_{2}N_{\text{t}}^{3.5}\right)+TI_{2}N_{\text{t}}^{3.5}\right)$.

\section{Two-Timescale ISAC-enabled Secure Communications Framework for Single-Waveguide PASS}
In this section, we consider the scenario where the BS is deployed with a single waveguide for signal transmission to Bob, while acquiring the CSI of Eve via echo signals to multiple LCXs.  Since a single waveguide provides only one effective spatial dimension, the resulting spatial channel is rank-one, which fundamentally restricts the transmission of at most one independent data stream. To address this issue, we design the two-timescale ISAC-enabled secure communications framework as shown in Fig.~\ref{fig:frame-single}. Specifically, each time block contains separate communication and sensing phases\footnote{We assume that the time duration occupied by the sensing phase does not exceed one CPI. The time allocation problem for communication and sensing phases is out of the scope of this work, which deserves further study in future works.}. On the one hand, the large-timescale pinching beamforming is carried out at the beginning of the $l$-th time block for the maximization of the average secrecy rate in the communication phase, which is also subject to the minimum echo signal power in the sensing phase. On the other hand, the small-timescale baseband power allocation can be updated across different CPIs.

\begin{figure}[t]
	\centering
	\includegraphics[scale=0.3]{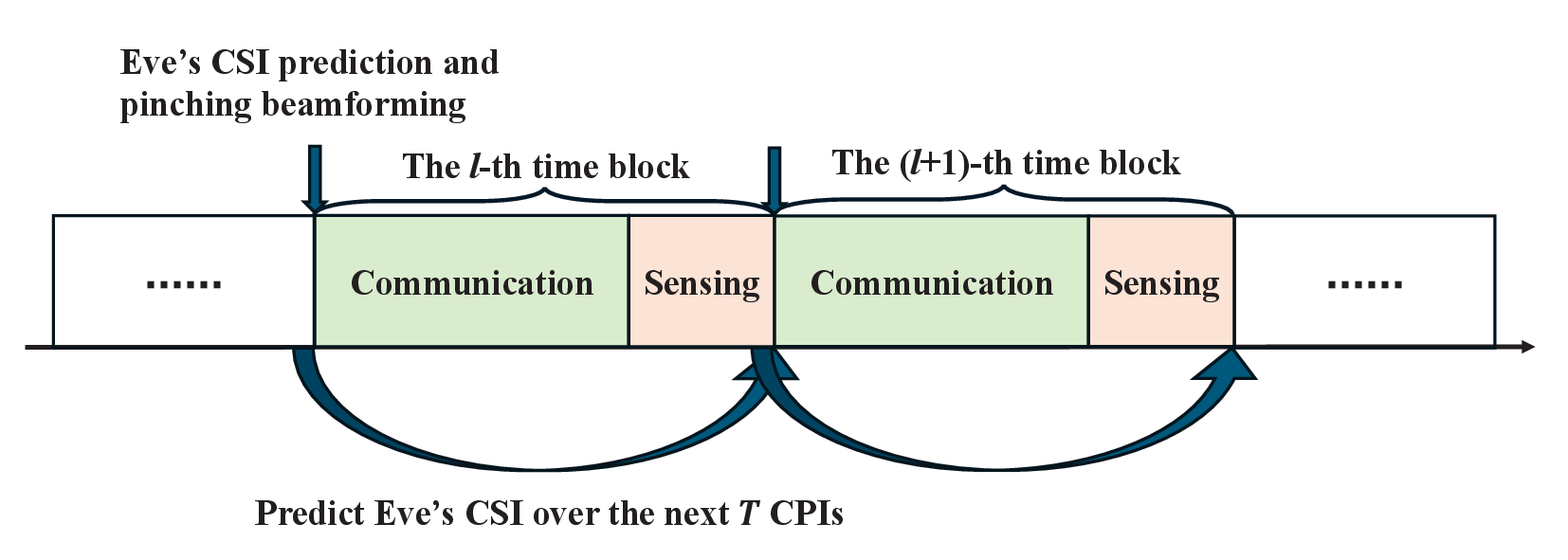}
	\caption{Illustration of the proposed two-timescale ISAC-enabled secure communications framework in single-waveguide PASS.}
	\label{fig:frame-single} 
\end{figure}

\subsection{Signal Model}
Let $s_{\mathrm{p}}$ denote the unit-power probing signal, then the received echo signal on the $q$-th LCX in the $t$-th CPI is given by
\begin{align}
y_{\mathrm{c}}^{q,t}=&\sqrt{p_\mathrm{s}}\sqrt{\beta}\left(\mathbf{v}^{q}\right)^T\left(\mathbf{h}_{\mathrm{e,r}}^{q,t}\left(\mathbf{v}_\mathrm{e}^{t},\boldsymbol{\psi}_\mathrm{e}^{t}\right)\right)^{T}\nonumber\\
&\mathbf{h}_{\mathrm{e,t}}^{t}\left(\mathbf{v}_\mathrm{e}^{t},\boldsymbol{\psi}_\mathrm{e}^{t},\mathbf{x}_\mathrm{p}\right)\mathbf{g}\left(\mathbf{x}_\mathrm{p}\right)s_{\mathrm{p}}+n_\mathrm{c}^{q,t}. 
\end{align}
where $p_\mathrm{s}$ denotes the transmit power during the sensing phase, satisfying $p_\mathrm{s}\leq P_{\max}$. Therefore, the echo signal observed at the BS is given by
\begin{align}
&\mathbf{y}_{\mathrm{c}}^{t}=[y_{\mathrm{c}}^{1,t},...,y_{\mathrm{c}}^{N_{\text{r}},t}]^{T}=\sqrt{p_\mathrm{s}}\sqrt{\beta}\mathbf{V}^{T}\left(\mathbf{h}_{\mathrm{e,r}}^{t}\left(\mathbf{v}_\mathrm{e}^{t},\boldsymbol{\psi}_\mathrm{e}^{t}\right)\right)^{T}\nonumber\\
&\mathbf{h}_{\mathrm{e,t}}^{t}\left(\mathbf{v}_\mathrm{e}^{t},\boldsymbol{\psi}_\mathrm{e}^{t},\mathbf{x}_\mathrm{p}\right)\mathbf{g}\left(\mathbf{x}_\mathrm{p}\right)s_{\mathrm{p}}+\mathbf{n}_{\mathrm{c}}^{t}.
\end{align}
Similar to the EKF scheme proposed in Section~\ref{subsec:EKF}, the Eve's state can be tracked with the echo signals $\mathbf{y}_{\mathrm{c}}^{t}$ in the single-waveguide case, for which the details are omitted here.

The received signal at Bob can be given by
\begin{align}
\label{eq:single-signal-Bob}
y_\text{b}^{t}=\sqrt{p_\mathrm{c}^{t}}\mathbf{h}_{\text{b}}\left(\mathbf{x}_{\text{p}}\right)\mathbf{g}\left(\mathbf{x}_{\text{p}}\right)s+n_\mathrm{b},
\end{align}
where $p_\mathrm{c}^{t}$ denotes the transmit power during the communication phase, satisfying $p_\mathrm{c}^{t}\leq P_{\max}$. Accordingly, the achievable data rate at Bob is given by
\begin{align}
\label{eq:single-rate-Bob}
R_{\text{b}}^{t}\left(\mathbf{x}_{\text{p}},p_\mathrm{c}^{t}\right)=\log_{2}\left(1+\frac{p_\mathrm{c}^{t}\left|\mathbf{h}_{\text{b}}\left(\mathbf{x}_{\text{p}}\right)\mathbf{g}\left(\mathbf{x}_{\text{p}}\right)\right|^2}{\sigma_\mathrm{b}^2}\right).
\end{align}
Similarly, the leakage information rate at Eve is given by
\begin{align}
\label{eq:single-rate-Eve}
R_{\text{e}}^{t}\left(\mathbf{x}_{\text{p}},p_\mathrm{c}^{t}\right)=\log_{2}\left(1+\frac{p_\mathrm{c}^{t}\left|\mathbf{h}_{\text{e,t}}^{t}\left(\mathbf{x}_{\text{p}}\right)\mathbf{g}\left(\mathbf{x}_{\text{p}}\right)\right|^2}{\sigma_\mathrm{e}^2}\right).
\end{align}
Therefore, the secrecy rate in the $t$-th CPI is given by $R_{\text{s}}^{t}=\left[R_{\text{b}}^{t}-R_{\text{e}}^{t}\right]^+$.

\subsection{Problem Formulation}
We aim to maximize the average secrecy rate in each time block, i.e., $T$ CPIs, by jointly optimizing the transmit power $p_\mathrm{s},\{p_\mathrm{c}^{t}\}_{t=1}^{T}$ and the pinching beamforming $\mathbf{x}_{\text{p}}$. The optimization problem can be formulated as 
\begin{subequations}
\label{eq:optimization-problem-single}
\begin{equation}
\label{eq:objective-function-single}
\max_{\mathbf{x}_{\text{p}},p_\mathrm{s},\{p_\mathrm{c}^{t}\}_{t=1}^{T}}\frac{1}{T}\sum_{t=1}^{T} R_{\text{s}}^{t},
\end{equation}
\begin{equation}
\label{eq:maximum-power-single}
{\rm{s.t.}} \ \
p_\mathrm{s}\leq P_{\max},p_\mathrm{c}^{t}\leq P_{\max},\forall t\in\left\{1, \dots, T-1\right\},
\end{equation}
\begin{equation}
\label{eq:echo-signal-power-single}
p(\mathbf{x}_{\text{p}},p_\mathrm{s})\geq\Gamma,
\end{equation}
\begin{equation}
\label{eq:position_constraint1-single}
0\leq x_{\text{p}}^{m}\leq L,\forall m,
\end{equation}
\begin{equation}
\label{eq:position_constraint2-single}
x_{\mathrm{p}}^{m+1}-x_{\mathrm{p}}^{m}\geq\Delta,\forall m.
\end{equation}
\end{subequations}
Constraint~\eqref{eq:maximum-power-single} ensures the BS transmit power does not exceed the threshold $P_{\max}$. In constraint~\eqref{eq:echo-signal-power-single}, $p(\mathbf{x}_{\text{p}})$ denotes the power of the received echo signal, given by 
\begin{align}
\begin{aligned}
\label{eq:echo-power-single}
p(\mathbf{x}_{\text{p}},p_\mathrm{s})\triangleq
p_\mathrm{s}\left(\mathbf{h}_{\mathrm{e}}^{t}\left(\mathbf{v}_{\mathrm{e}}^{t},\boldsymbol{\psi}_{\text{e}}^{t},\mathbf{x}_{\text{p}}\right)\right)^{H}\mathbf{h}_{\mathrm{e}}^{t}\left(\mathbf{v}_{\mathrm{e}}^{t},\boldsymbol{\psi}_{\text{e}}^{t},\mathbf{x}_{\text{p}}\right),
\end{aligned}
\end{align}
where the round-trip channel vector $\mathbf{h}_{\text{e}}^{t}\left(\mathbf{v}_{\mathrm{e}}^{t},\boldsymbol{\psi}_{\text{e}}^{t},\mathbf{x}_\mathrm{p}\right)$ is given by
\begin{align}
&\mathbf{h}_{\text{e}}^{t}\triangleq\sqrt{\beta}\mathbf{V}^{T}\left(\mathbf{h}_{\text{e,r}}^{t}\left(\mathbf{v}_{\mathrm{e}}^{t},\boldsymbol{\psi}_{\text{e}}^{t}\right)\right)^{T}\mathbf{h}_{\text{e,t}}^{t}\left(\mathbf{v}_{\mathrm{e}}^{t},\boldsymbol{\psi}_{\text{e}}^{t},\mathbf{x}_\mathrm{p}\right)\mathbf{g}\left(\mathbf{x}_\mathrm{p}\right).
\end{align}
Note that the optimization variable $p_{\text{s}}$ exists only in constraints~\eqref{eq:maximum-power-single} and~\eqref{eq:echo-signal-power-single}, and is irrelevant to the objective function in~\eqref{eq:objective-function-single}. Therefore, we set $p_\mathrm{s}=P_{\max}$ to yield the largest feasible region for the optimization of $\mathbf{x}_{\text{p}}$.

\subsection{Proposed Solution}
According to Eq.~\eqref{eq:objective-function-multiple}, the secrecy rate in the $t$-th CPI can be further expressed as follows
\begin{align}
R_{\text{s}}^{t}&=\left[\log_{2}\left(\frac{p_{\mathrm{c}}^{t}\left|\mathbf{h}_{\text{b}}\left(\mathbf{x}_{\text{p}}\right)\mathbf{g}\left(\mathbf{x}_{\text{p}}\right)\right|^2+\sigma_\mathrm{b}^2}{p_{\mathrm{c}}^{t}\left|\mathbf{h}_{\text{e,t}}^{t}\left(\mathbf{x}_{\text{p}}\right)\mathbf{g}\left(\mathbf{x}_{\text{p}}\right)\right|^2+\sigma_\mathrm{e}^2}\right)\right]^+.
\end{align}
\begin{lemma}
For a given pinching beamforming vector $\mathbf{x}_{\text{p}}$, if Bob has an effective channel quality advantage over Eve in the $t$-th CPI, i.e.,
\begin{align}
\label{eq:channel-condition}
\frac{\left|\mathbf{h}_{\text{b}}\left(\mathbf{x}_{\text{p}}\right)\mathbf{g}\left(\mathbf{x}_{\text{p}}\right)\right|^2}{\sigma_{\mathrm b}^2} > \frac{\left|\mathbf{h}_{\text{e,t}}^{t}\left(\mathbf{x}_{\text{p}}\right)\mathbf{g}\left(\mathbf{x}_{\text{p}}\right)\right|^2}{\sigma_{\mathrm e}^2},
\end{align}
the secrecy rate $R_{\text{s}}^{t}$ is monotonically increasing with $p_{\mathrm{c}}^{t}$.
\end{lemma}
\begin{proof}
For notational simplicity, define $a_{\text{b}}=\left|\mathbf{h}_{\text{b}}\left(\mathbf{x}_{\text{p}}\right)\mathbf{g}\left(\mathbf{x}_{\text{p}}\right)\right|^2$ and $a_{\text{e}}^{t}=\left|\mathbf{h}_{\text{e,t}}^{t}\left(\mathbf{x}_{\text{p}}\right)\mathbf{g}\left(\mathbf{x}_{\text{p}}\right)\right|^2$. Then, the secrecy rate in the $t$-th CPI can be rewritten as
\begin{align}
R_{\text{s}}^{t}\left(p_{\mathrm{c}}^{t}\right) = \left[ \log_2\left(1+\frac{p_{\mathrm{c}}^{t}a_{\text{b}}} {\sigma_{\mathrm b}^{2}}\right) - \log_2\left(1+\frac{p_{\mathrm{c}}^{t}a_{\text{e}}^{t}} {\sigma_{\mathrm e}^{2}}\right) \right]^+.
\end{align}
When $\frac{a_{\text{b}}} {\sigma_{\mathrm b}^{2}} > \frac{a_{\text{e}}^{t}} {\sigma_{\mathrm e}^{2}}$, the term inside $[\cdot]^+$ is positive for $p_{\mathrm{c}}^{t}>0$. Thus, the derivative of $R_{\text{s}}^{t}$ with respect to $p_{\mathrm{c}}^{t}$ is given by
\begin{align}
\frac{\partial R_{\mathrm s}^{t}}{\partial p_{\mathrm{c}}^{t}} &= \frac{1}{\ln 2} \left( \frac{a_{\text{b}}}{\sigma_{\mathrm b}^{2}+p_{\mathrm{c}}^{t}a_{\text{b}}} - \frac{a_{\text{e}}^{t}}{\sigma_{\mathrm e}^{2}+p_{\mathrm{c}}^{t}a_{\text{e}}^{t}}\right)\nonumber\\
&=\frac{1}{\ln 2}\frac{a_{\text{b}}\sigma_{\mathrm e}^{2}-a_{\text{e}}^{t}\sigma_{\mathrm b}^{2}}{\left(\sigma_{\mathrm b}^{2}+p_{\mathrm{c}}^{t}a_{\text{b}}\right)\left(\sigma_{\mathrm e}^{2}+p_{\mathrm{c}}^{t}a_{\text{e}}^{t}\right)}. 
\end{align}
Since $\frac{a_{\text{b}}} {\sigma_{\mathrm b}^{2}} > \frac{a_{\text{e}}^{t}} {\sigma_{\mathrm e}^{2}}$, we have $a_{\text{b}}\sigma_{\mathrm e}^{2}-a_{\text{e}}^{t}\sigma_{\mathrm b}^{2} > 0$. Therefore, the secrecy rate $R_{\text{s}}^{t}$ is monotonically increasing with $p_{\mathrm{c}}^{t}$.
\end{proof}

\begin{algorithm}[tp]
	\caption{Proposed Two-Timescale Optimization Algorithm for Solving Problem~\eqref{eq:optimization-problem-single}}
        \label{alg:single}
    \LinesNumbered
    \KwIn{Bob's location $\boldsymbol{\psi}_{\text{b}}$}
    \textbf{At the beginning of each time block (Large-timescale pinching beamforming):}\\
    Predict Eve’s mobility states over the next $T$ CPIs $\{\hat{\boldsymbol{\xi}}^{t}\}_{t=1}^{T}$ with the prediction model given in~\eqref{eq:prior-prediction};\\
    Set $p_\mathrm{s}=P_{\max}$ and $p_\mathrm{c}^{t}=P_{\max},\forall t\in\left\{1, \dots, T-1\right\}$;\\
	Initialize the pinching beamforming $\mathbf{x}_{\text{p}}$;\\
    Update $\mathbf{x}_{\text{p}}$ by the element-wise one-dimensional searching algorithm;\\
    \KwOut{Pinching beamforming $\mathbf{x}_{\text{p}}^{*}$ for the current time block;}
    \textbf{In each CPI of the time block (Small-timescale power allocation):}\\
    \For{$t=1$ \KwTo $T-1$}
    {
    \If{inequality~\eqref{eq:channel-condition} holds}
    {
    Set $p_{\mathrm{c}}^{t}=P_{\max}$\;
    }
    \Else
    {
    Set $p_{\mathrm{c}}^{t}=0$\;
    }
    }
    \KwOut{$p_{\mathrm{c}}^{t}$ for each CPI.}
\end{algorithm}

Note that, in the large-timescale pinching beamforming stage, we can make the inequality in~\eqref{eq:channel-condition}, that is derived based on the predicted Eve's positions, holds by properly adjusting $\mathbf{x}_{\text{p}}$. Therefore, by setting $p_\mathrm{c}^{t}=P_{\max},\forall t$ at the pinching beamforming stage, we can formulate  the optimization subproblem with respect to $x_{\text{p}}^{m}$ as follows:
\begin{subequations}
\label{eq:optimization-problem-single-X}
\begin{equation}
\label{eq:objective-function-single-X}
\max_{x_{\text{p}}^{m}}\frac{1}{T}\sum_{t=1}^{T} \log_{2}\left(\frac{P_{\max}\left|A_{\text{b}}^{t}\left(x_{\text{p}}^{m}\right)\right|^2+M_{\text{t}}\sigma_\mathrm{b}^2}{P_{\max}\left|A_{\text{e}}^{t}\left(x_{\text{p}}^{m}\right)\right|^2+M_{\text{t}}\sigma_\mathrm{e}^2}\right),
\end{equation}
\begin{equation}
\label{eq:echo-signal-power-single-X}
{\rm{s.t.}} \ \
p\left(x_{\text{p}}^{m}\right)\geq\Gamma,
\end{equation}
\begin{equation}
\label{eq:position_constraint-single-X}
\eqref{eq:position_constraint1-single},\eqref{eq:position_constraint2-single},
\notag
\end{equation}
\end{subequations}
where $A_{\text{b}}^{t}\left(x_{\text{p}}^{m}\right)$ and $A_{\text{e}}^{t}\left(x_{\text{p}}^{m}\right)$ represent the effective signal terms that depend on $x_{\text{p}}^{m}$, given by
\begin{subequations}
\begin{equation}
A_{\text{b}}^{t}\left(x_{\text{p}}^{m}\right)=\frac{\eta e^{-j\theta^{m}}}{\left\|\psi_{\text{b}}-\psi_{\mathrm{p}}^{m}\right\|}+\sum_{q\neq m}^{M_{\text{t}}}\frac{\eta e^{-j\theta^{q}}}{\left\|\psi_{\text{b}}-\psi_{\mathrm{p}}^{q}\right\|},
\end{equation}
\begin{equation}
A_{\text{e}}^t\left(x_{\text{p}}^{m}\right)=\frac{\eta e^{-j\phi^{t,m}}}{\left\|\psi_{\text{e}}^t-\psi_{\mathrm{p}}^{m}\right\|}+\sum_{q\neq m}^{M_{\text{t}}}\frac{\eta e^{-j\phi^{t,q}}}{\left\|\psi_{\text{e}}^t-\psi_{\mathrm{p}}^{q}\right\|},
\end{equation}
\end{subequations}
with $\theta^{m}=\frac{2\pi}{\lambda}\left\|\boldsymbol{\psi}_{\text{b}}-\boldsymbol{\psi}^{m}_{\text{p}}\right\|+\frac{2\pi }{\lambda_{\text{g}}}x_{\text{p}}^{m}$ and $\phi^{t,m}=\frac{2\pi}{\lambda}\left(\left\|\boldsymbol{\psi}_{\text{e}}^{t}-\boldsymbol{\psi}^{m}_{\text{p}}\right\|+\Delta Tv^{t,m}\right)+\frac{2\pi }{\lambda_{\text{g}}}x_{\text{p}}^{m}$. Again, we employ the element-wise method for the pinching beamforming optimization. It can be observed that problem~\eqref{eq:optimization-problem-single-X} is a single-variable optimization problem under the continuous deployment and spacing constraints in~\eqref{eq:position_constraint1-single} and~\eqref{eq:position_constraint2-single}, which can be solved by the one-dimensional search over the feasible set $\mathcal{S}_{x}$ given in Eq.~\eqref{eq:feasible-set}.

Note that, due to the prediction error of Eve's mobility states in the large-timescale pinching beamforming stage, the condition in~\eqref{eq:channel-condition} may not hold with the refined Eve's CSI in each CPI. Therefore, we need to check again whether the condition~\eqref{eq:channel-condition} is valid for each CPI. If yes, we set $p_{\mathrm{c}}^{t}=P_{\max}$. Otherwise, we set $p_{\mathrm{c}}^{t}=0$, given that no positive secrecy rate can be guaranteed for confidential information transmission in the current CPI. The complete algorithm for solving the original problem~\eqref{eq:optimization-problem-single} is summarized in \textbf{Algorithm~\ref{alg:single}}. Specifically, the pinching beamforming is optimized only at the beginning of each time block. The computational complexity of \textbf{Algorithm~\ref{alg:single}} is given by $\mathcal{O}\left(I_{1}NTN_{\text{t}}M_{\text{t}}\right)$.


\section{Simulation Results}
In this section, numerical results are provided to verify the
effectiveness of the proposed two-timescale ISAC-enabled secure communications in PASS, for both multiple- and single-waveguide scenarios.
\vspace{-4mm}
\subsection{Simulation Setup}
Unless otherwise specified, the carrier frequency, the effective refractive index of the waveguide, the maximum transmit power at the BS, and the noise powers at Bob and Eve are set to $f_{\mathfrak{c}}=6~\mathrm{GHz}$, $n_{\mathrm{eff}}=1.4$, $P_{\mathrm{max}}=40~\mathrm{dBm}$ and $\sigma_\mathrm{b}^2=\sigma_\mathrm{e}^2=-90~\mathrm{dBm}$, respectively. For the multiple-waveguide scenario, the BS is deployed with $N_{\text{t}}=2$ waveguides, each connected to a dedicated RF chain and equipped with $M_{\text{t}}=8$ PAs. For echo reception, the BS employs $N_{\text{r}}=2$ LCXs, whose slots are placed at a uniform interval of $1~\mathrm{m}$. For the geometric deployment, the waveguides and LCXs, each with length $L = 10~\mathrm{m}$, are arranged in parallel at the height of $d = 3~\mathrm{m}$. The inter-waveguide spacing and inter-LCX spacing are both fixed at $5~\mathrm{m}$, and each waveguide is separated from its neighboring LCX by $0.5~\mathrm{m}$. Moreover, we set the CPI duration as $\Delta T=100~\mathrm{ms}$ and adopt a constant RCS of Eve as $\beta = 1$~\cite{9246715}. The sensing threshold is set as $\Gamma=-70~\mathrm{dBm}$. The PAs positions are updated once every four CPIs, i.e., $T=4$. The position of Bob and the initial position of Eve are set to $[4,0,0]~\mathrm{m}$ and $[2,2,0]~\mathrm{m}$,  respectively. The initial velocities of Eve are set to $v_x=3~\mathrm{m/s}$ and $v_y=3\mathrm{~m/s}$ with velocity variance $\sigma_{v_x}^2=0.01~\mathrm{(m/s)^2}$ and $\sigma_{v_y}^2=0.02~\mathrm{(m/s)^2}$, respectively. 
\vspace{-4mm}
\subsection{Multiple-Waveguide Scenario}
We first present simulation results in the multiple-waveguides scenario. To demonstrate the effectiveness of the proposed scheme, we consider the following baselines for performance comparisons. 
\begin{itemize}
    \item \textbf{MIMO (Fully-digital beamforming)}: In this benchmark, the BS employs a uniform linear array centered at $[0,0,d]$ and aligned with the $x$-axis. The array is composed of $N_{\text{t}}$ antennas separated by $\lambda/2$, each connected to a dedicated RF chain, such that fully digital beamforming is available at the BS. The resultant average secrecy rate maximization problem can be solved by applying the method detailed in Section~\ref{subsec:baseband-beamforming}.
    \item \textbf{Massive MIMO (Hybrid beamforming)}: In this benchmark, a uniform linear array is deployed with its center located at $[0,0,d]$. The BS is equipped with $N_{\text{t}}M_{\text{t}}$ antennas connected to $N_{\text{t}}$ RF chains through a sub-connected structure, i.e., the hybrid beamforming is enabled at the BS.  The resultant secure communication problem is solved by applying the algorithm in~\cite{7397861}.
\end{itemize}

\begin{figure}
    \centering
    \begin{subfigure}{\linewidth}
        \centering
        \includegraphics[scale=0.5]{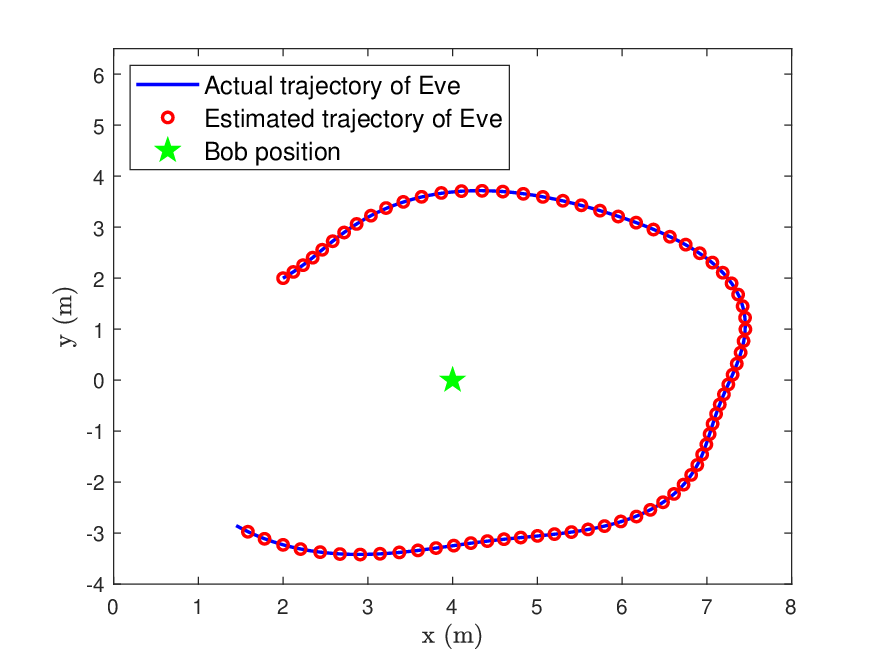}
        \caption{Tracking results of Eve’s trajectory.}
        \label{fig:trajectory-multi}
    \end{subfigure}
    \begin{subfigure}{\linewidth}
        \centering
        \includegraphics[scale=0.5]{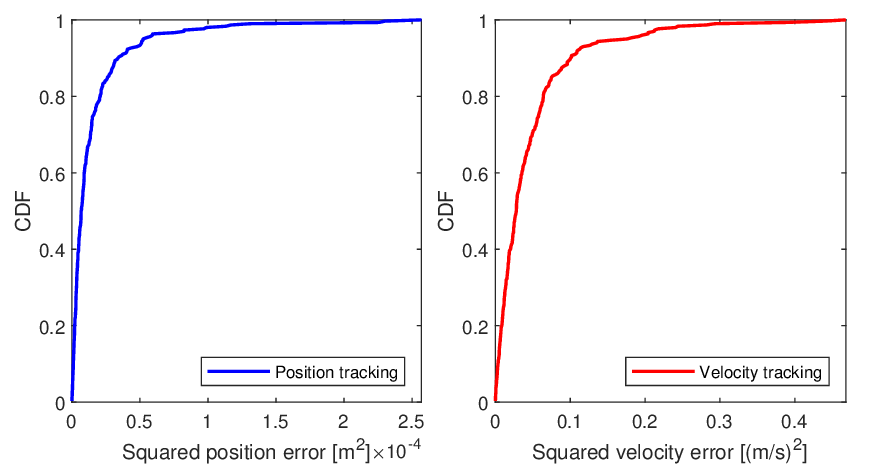}
        \caption{CDF of the squared position and velocity tracking errors.}
        \label{fig:cdf-multi}
    \end{subfigure}
    \caption{Illustration of the Eve's states tracking results in the multiple-waveguide scenario.}
    \label{fig:tracking-multi}
\end{figure}

Fig.~\ref{fig:tracking-multi} demonstrates the Eve's state tracking performance of the proposed framework. It can be observed from Fig.~3(a) that, the estimated trajectory of Eve perfectly coincides with the ground truth, which underscores the superiority of the considered PASS in locations and velocities sensing. Notably, the enlarged aperture of LCXs extends the near-field region, thereby enabling full-dimensional tracking by exploiting spherical waves. To provide a more comprehensive evaluation of the mobility state tracking performance, Fig.~3(b) further presents the empirical cumulative distribution function (CDF) of the squared position and velocity tracking errors. As can be observed, the order of magnitude of the position tracking error is much smaller than that of the velocity tracking error. This is because, the velocities variances $\Delta v_{x}$ and $\Delta v_{y}$ can not be captured in the prior prediction model, which restricts the subsequent data-fusion performance. However, this performance loss is alleviated for the position estimation, as the velocities variances give less impact on the prior position prediction under a small time period of one CPI. 

\begin{figure}[t]
	\centering
	\includegraphics[scale=0.58]{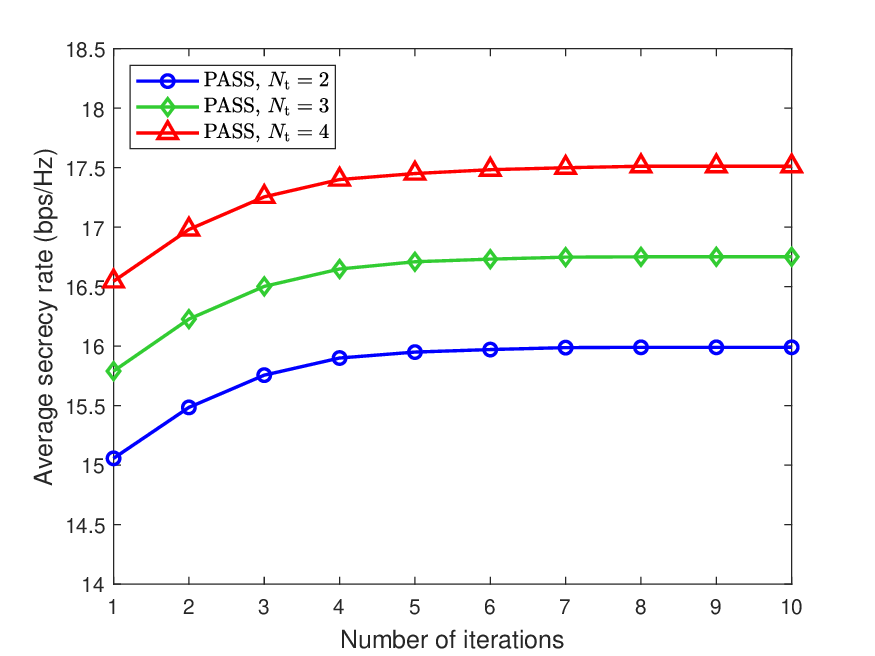}
	\caption{Convergence behavior of the proposed AO algorithm for jointly optimizing pinching and baseband beamforming.}
	\label{fig:rate-iter-multi} 
\end{figure}

In Fig.~\ref{fig:rate-iter-multi}, we present the convergence behavior of the proposed AO algorithm for jointly optimizing the pinching and baseband beamforming under different waveguide numbers $N_{\text{t}}$. As observed, the average secrecy rate improves significantly in the first few iterations and stabilizes after about five iterations. This convergence behavior suggests that the proposed two-timescale optimization algorithm can be implemented with relatively low iterative overhead in practical systems. It can be further observed that increasing the number of waveguides leads to a higher average secrecy rate. This is because, on the one hand, more waveguides bring in higher channel gain for sending the AN to Eve, so as to improve the subsequent beamforming effectiveness with higher Eve's state tracking accuracy. On the other hand, more waveguides also indicate improved spatial degree of freedoms (DoFs) for steering the communication signal towards Bob and reducing the information leakage to Eve.  

\begin{figure}[t]
	\centering
	\includegraphics[scale=0.58]{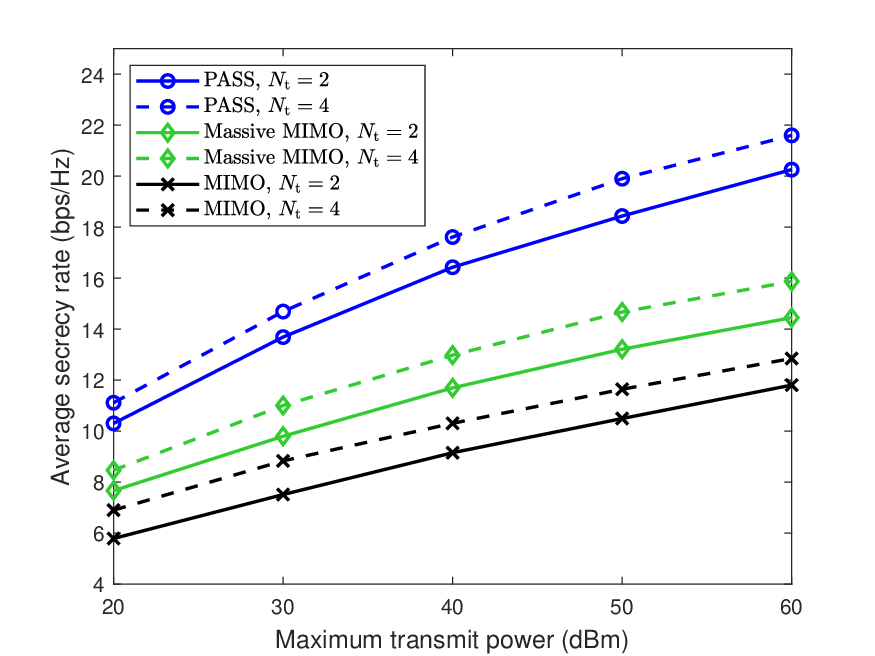}
	\caption{Average secrecy rate versus BS maximum transmit power.}
	\label{fig:rate-p-multi}
\end{figure}

Fig.~\ref{fig:rate-p-multi} depicts the average secrecy data rate versus the BS maximum transmit power $P_{\mathrm{max}}$, where we set $N_{\text{t}}=\{2,4\}$. As expected, the average secrecy rate increases monotonically with the transmit power across all considered schemes. Notably, the considered PASS significantly outperforms the conventional MIMO and massive MIMO systems. This is expected, as PASS can provide substantial performance gains by effectively mitigating large-scale path loss. For instance, when $N_{\text{t}}=2$ and $P_{\mathrm{max}}=20~\mathrm{dBm}$, the average secrecy rate obtained by PASS is about $10.3~\mathrm{bps/Hz}$, which is notably higher than those of the conventional MIMO and massive MIMO, i.e., $5.8~\mathrm{bps/Hz}$ and $7.6~\mathrm{bps/Hz}$, respectively.

\begin{figure}[t]
	\centering
	\includegraphics[scale=0.58]{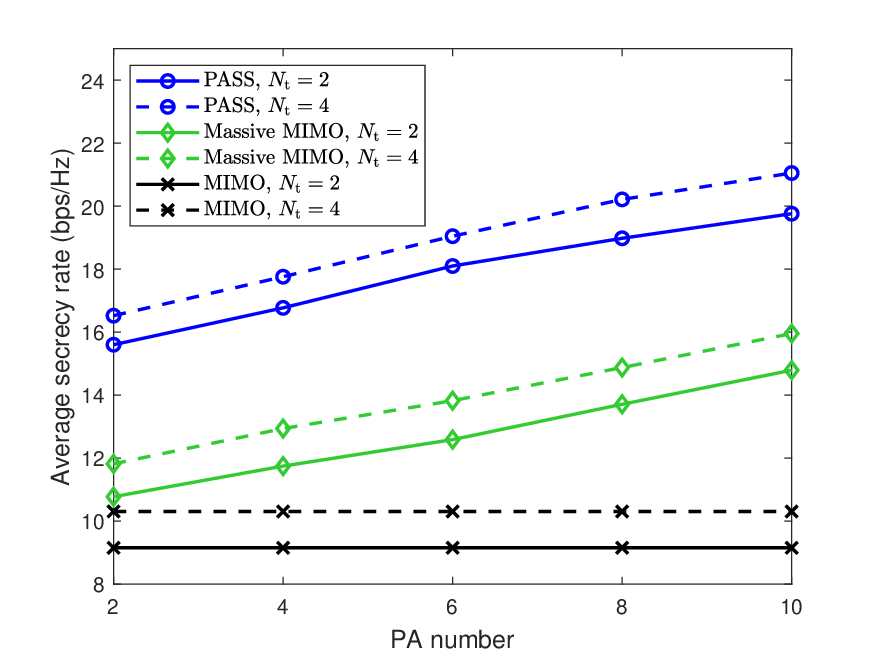}
	\caption{Average secrecy rate versus the number of PAs.}
	\label{fig:rate-pa-multi} 
\end{figure}

Fig.~\ref{fig:rate-pa-multi} illustrates the impact of the number of PAs per waveguide $M_{\text{t}}$ on the average secrecy rate under $P_{\mathrm{max}}=40~\mathrm{dBm}$. As shown in the figure, the average secrecy rate improves as $M_{\text{t}}$ increases. For instance, as the number of PAs on each waveguide grows from 2 to 10, the average secrecy rate of the PASS increases by around $26~\%$. The reason is that, increasing $M_{\text{t}}$ introduces additional spatial DoFs, which can be exploited to enhance the secrecy performance. It is also worthy noting that, the hybrid MIMO generally requires a larger number of phase shifters for analog beamforming with the increment of $M_{\text{t}}$. In contrast, the secrecy performance improvement achieved by PASS is realized under a low-cost deployment architecture, underscoring its economic practicality.

\vspace{-2mm}
\subsection{Single-Waveguide Scenario}
For the single-waveguide scenario, we consider the following benchmarks for demonstrating the effectiveness of our proposed scheme. 
\begin{itemize}
    \item \textbf{Random position}: In this benchmark, PAs are placed in random positions along the waveguide, given that the minimum spacing of $\lambda/2$, is satisfied. The final secrecy performance is evaluated by averaging over 500 independent realizations.
    \item \textbf{Conventional MIMO (Analog beamforming)}: In this benchmark, the BS is equipped with $N_{\text{t}}$ antennas, while only one RF chain is activated. The antennas are centered at $[0,0,d]$ and uniformly arranged along the $x$-axis with half-wavelength spacing. The corresponding analog beamforming is designed based on the maximum ratio transmission (MRT).
\end{itemize}

\begin{figure}
    \centering
    \begin{subfigure}{\linewidth}
        \centering
        \includegraphics[scale=0.5]{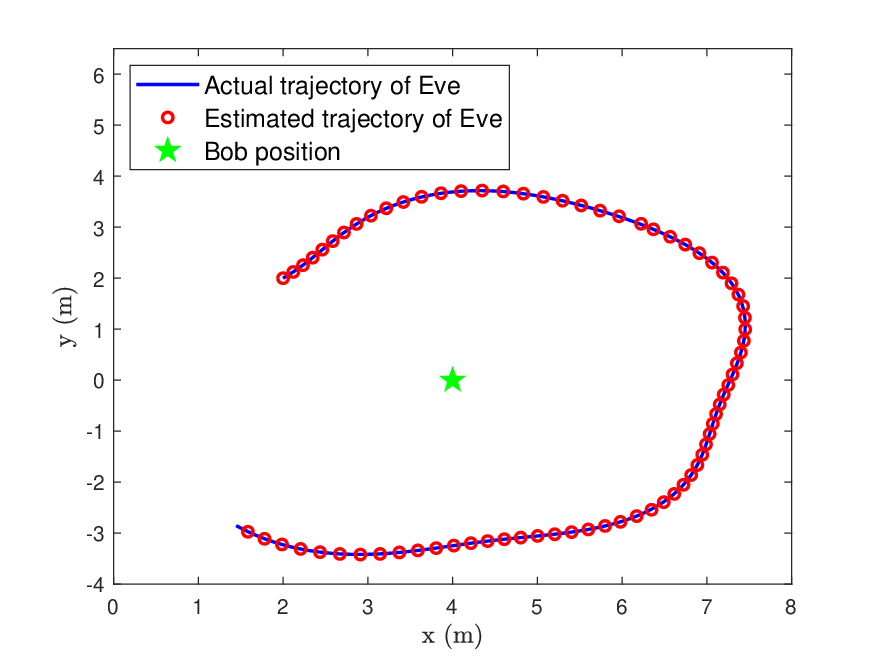}
        \caption{Tracking results on Eve’s trajectory.}
        \label{fig:trajectory-single}
    \end{subfigure}
    \begin{subfigure}{\linewidth}
        \centering
        \includegraphics[scale=0.5]{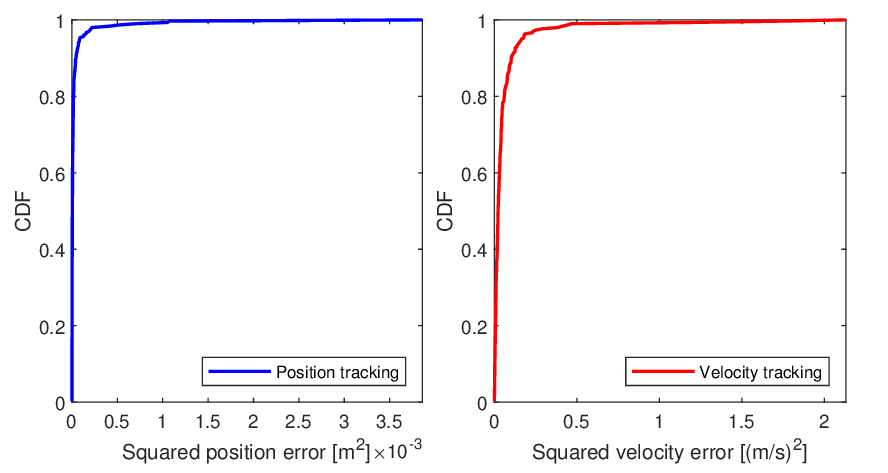}
        \caption{CDF of the squared position and velocity tracking errors.}
        \label{fig:cdf-single}
    \end{subfigure}
    \caption{Illustration of the Eve’s states tracking results in the single-waveguide scenario.}
    \label{fig:tracking-single}
\end{figure}

Fig.~\ref{fig:tracking-single} demonstrates the Eve's states tracking results, with $P_{\mathrm{max}}=40~\mathrm{dBm}$ and $M_{\text{t}}=8$. As shown in Fig.~7(a), the proposed EKF-based method can also provide accurate position tracking for Eve in the single-waveguide scenario. Fig.~7(b) illustrates the empirical CDF of the squared position and velocity tracking errors. As can be seen from the figure, the position tracking error remains on the order of $10^{-3}$, while the velocity tracking error is around the $10^{-1}$ level, validating the effectiveness of the proposed EKF-based tracking method in the single-waveguide scenario.

\begin{figure}[t]
	\centering
	\includegraphics[scale=0.58]{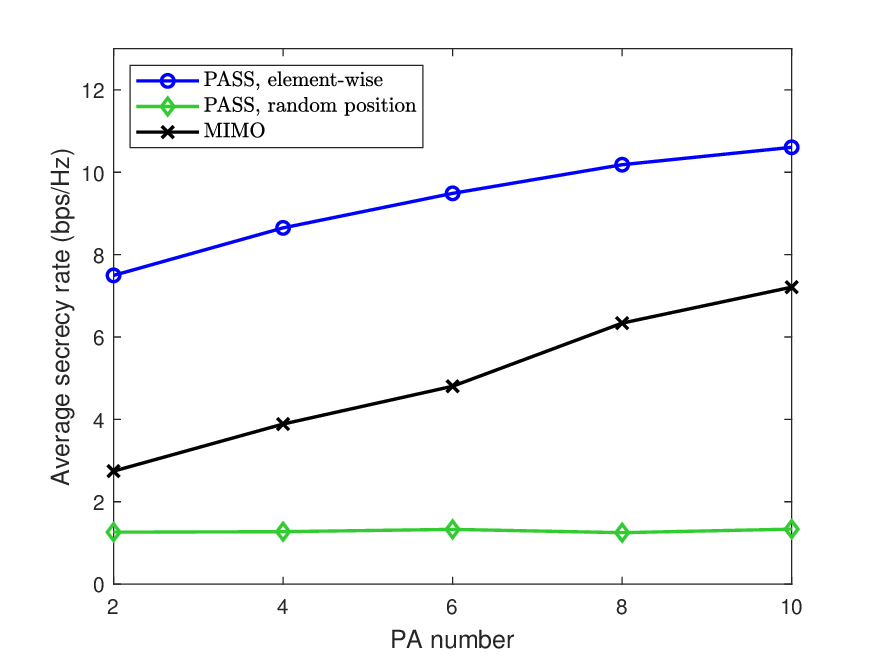}
	\caption{Average secrecy rate versus the number of PAs.}
	\label{fig:rate-pa-single} 
\end{figure}


Fig.~\ref{fig:rate-pa-single} illustrates the average secrecy rate versus the number of PAs. As can be observed, except
for the random position scheme, all methods achieve improved average secrecy rate performance as $M_{\text{t}}$ increases. For the random-position scheme, increasing the number of PAs does not necessarily yield a notable improvement in the secrecy rate, since their random spatial deployment gives rise to unpredictable channel conditions. In contrast, the other schemes exploit the increase in $M_{\text{t}}$ to obtain more spatial DoFs, which further enlarge the channel quality differences and thus improve the average secrecy rate. These observations highlight the importance of optimizing PAs positions in PASS. Moreover, it is also worth noting that, the PASS significantly outperforms the conventional MIMO system with analog beamforming. This is expected, as PASS can significantly improve the channel quality difference between  Bob and Eve by flexibly moving PAs. 

\section{Conclusions}
In this paper, a novel two-timescale optimization framework was proposed for PASS-enabled secure ISAC, where the BS employed PAs and LCXs for the signal transmission and the echo signal reception, respectively. The moving-Eve's mobility states were tracked by employing the EKF method that fused the predicted Eve’s states and the measured ones. Considering the limited PAs positions reconfiguration overhead, the pinching beamforming and baseband processing were executed in the large and small timescales, respectively. Specifically, we first considered the multiple-waveguide scenario, where the BS was enabled to transmit AN together with communication signals. We formulated a joint pinching and baseband beamforming problem with the aim of maximizing the average secrecy rate. To solve the resultant problem across different timescales, the AO algorithm was first applied in the large timescale for determining the pinching beamforming with the predicted Eve’s mobility states. Afterwards, the baseband beamforming was continuously updated with the refined Eve’s states obtained by the real-time sensing. We then considered the single-waveguide scenario. Since at most one data stream could be carried, an ISAC frame structure containing separate communication and sensing phases was designed. The large-timescale pinching beamforming was optimized in an element-wise manner, while the small-timescale power allocation was adaptively determined according to the effective channel quality comparison between Bob and Eve. Simulation results demonstrated the effectiveness of the proposed two-timescale secure ISAC framework in terms of Eve’s states tracking accuracy as well as the secrecy rate enhancement.

\bibliographystyle{IEEEtran}
\bibliography{mybib}
\end{document}